%% file: cp25.tex
\documentclass[a4paper,USenglish,cleveref,thm-restate]{lipics-v2021}

\hideLIPIcs%

\usepackage{tikz}
\usetikzlibrary{arrows,shapes,automata,petri,positioning,fit}

\usepackage{subcaption}

\usepackage{amssymb}
\usepackage{amsmath}

\input{sections/macro.tex}
\title{Scalable Counting of Minimal Trap Spaces and Fixed Points in Boolean Networks} %

\author{Mohimenul Kabir\footnote{Corresponding author}}{School of Computing, National University of Singapore, Singapore}{e0546090@u.nus.edu}{https://orcid.org/0000-0001-7551-0337}{}

\author{Van-Giang Trinh}{Inria Saclay, EP Lifeware, Palaiseau, France}{van-giang.trinh@inria.fr}{https://orcid.org/0000-0001-6581-998X}{}%

\author{Samuel Pastva}{Faculty of Informatics, Masaryk University, Botanicka 68a, Brno, 60200, Czechia}{xpastva@fi.muni.cz}{https://orcid.org/0000-0003-1993-0331}{}

\author{Kuldeep S Meel}{School of Computer Science, Georgia Institute of Technology, Atlanta, USA}{meel@cs.toronto.edu}{https://orcid.org/0000-0001-9423-5270}{}

\authorrunning{Kabir et al.} %

\Copyright{Mohimenul Kabir, Van-Giang Trinh, Samuel Pastva, Kuldeep S Meel} %



\category{Regular Paper} %

\supplement{The extended version is available at: \url{https://arxiv.org/abs/2506.06013}.}%

\supplementdetails[subcategory={Source Code}, cite={}, swhid={}]{Software}{https://github.com/meelgroup/bn-counting}
\supplementdetails[subcategory={}, cite={}, swhid={}]{Dataset}{https://zenodo.org/records/15141045}

\funding{The work was supported by the National University of Singapore and the Natural Sciences and Engineering Research Council of Canada (NSERC) [RGPIN-$2024$-$05956$].}%

\acknowledgements{We are thankful to the reviewers for their feedback on the paper. The computational work of this paper was conducted in the Niagara cluster.}%

\nolinenumbers %

\EventEditors{Roland H. C. Yap}
\EventNoEds{1}
\EventLongTitle{29th International Conference on Principles and Practice of Constraint Programming (CP 2023)}
\EventShortTitle{CP 2023}
\EventAcronym{CP}
\EventYear{2023}
\EventDate{August 27--31, 2023}
\EventLocation{Toronto, Canada}
\EventLogo{}
\SeriesVolume{280}
\ArticleNo{31}
\begin{document}

\maketitle

\begin{abstract}
	Boolean Networks (BNs) serve as a fundamental modeling framework for capturing complex dynamical systems across various domains, including systems biology, computational logic, and artificial intelligence. 
	A crucial property of BNs is the presence of trap spaces --- subspaces of the state space that, once entered, cannot be exited. 
	Minimal trap spaces, in particular, play a significant role in analyzing the long-term behavior of BNs, making their efficient enumeration and counting essential.
	The fixed points in BNs are a special case of minimal trap spaces.
	In this work, we formulate several meaningful counting problems related to minimal trap spaces and fixed points in BNs.
	These problems provide valuable insights both within BN theory (e.g., in probabilistic reasoning and dynamical analysis) and in broader application areas, including systems biology, abstract argumentation, and logic programming. 
	To address these computational challenges, we propose novel methods based on {\em approximate answer set counting}, leveraging techniques from answer set programming.
	Our approach efficiently approximates the number of minimal trap spaces and the number of fixed points without requiring exhaustive enumeration, making it particularly well-suited for large-scale BNs.
	Our experimental evaluation on an extensive and diverse set of benchmark instances shows that our methods significantly improve the feasibility of counting minimal trap spaces and fixed points, paving the way for new applications in BN analysis and beyond.
\end{abstract}

\input{sections/introduction.tex}

\input{sections/preliminaries.tex}

\input{sections/related-work.tex}
\input{sections/problem_formulate.tex}

\input{sections/computation.tex}
\input{sections/experiment.tex}

\input{sections/conclusion.tex}

\bibliography{cp25}
\clearpage
\appendix

\input{sections/appendix-a.tex}

\input{sections/appendix-b.tex}

\input{sections/appendix-c.tex}
\input{sections/appendix-d.tex}

\input{sections/appendix-e.tex}

\end{document}

%% file: sections/macro.tex
\usepackage{subcaption}
\usepackage{booktabs}
\usepackage{xfrac}
\usepackage{xspace}
\usepackage{flexisym}

\newcommand{\cnt}{\ensuremath{\mathsf{cnt}}}

\newcommand{\Card}[1]{|#1|}
\newcommand{\var}[1]{\mathsf{Var}({#1})}

\newcommand{\stg}[1]{\mathsf{sstg}({#1})}
\newcommand{\atg}[1]{\mathsf{astg}({#1})}

\newcommand{\dng}[1]{{{\sim}#1}}
\newcommand{\head}[1]{\mathsf{h}({#1})}

\newcommand{\pbody}[1]{\mathsf{b}^+({#1})}
\newcommand{\nbody}[1]{\mathsf{b}^-({#1})}

\newcommand{\at}[1]{\mathsf{at}({#1})}

\newcommand{\as}[1]{\mathsf{AS}({#1})}

\newcommand{\acfirstfix}[1]{\text{C-FIX-1}}
\newcommand{\acfirstmts}[1]{\text{C-MTS-1}}
\newcommand{\acsecondfix}[1]{\text{C-FIX-2}}
\newcommand{\acsecondmts}[1]{\text{C-MTS-2}}
\newcommand{\acsecondfixp}[1]{\text{P-FIX-2}}
\newcommand{\acsecondmtsp}[1]{\text{P-MTS-2}}
\newcommand{\acthirdfix}[1]{\text{C-FIX-3}}
\newcommand{\acthirdmts}[1]{\text{C-MTS-3}}

\newcommand{\twod}{\ensuremath{\mathbb{B}}}
\newcommand{\threed}{\ensuremath{\mathbb{B}_{\star}}}

\newcommand{\fasp}{\ensuremath{\mathtt{fASP}}\xspace}
\newcommand{\tsconj}{\ensuremath{\mathtt{tsconj}}\xspace}
\newcommand{\bn}{\ensuremath{f}} %

\newcommand{\phen}{\ensuremath{\beta}}
\newcommand{\perm}{\ensuremath{\sigma}}
\newcommand{\pert}{\ensuremath{\mathcal{X}}}
\newcommand{\proj}{\ensuremath{\Delta}}
\newcommand{\aproj}{\ensuremath{\Omega}}
\newcommand{\phenasp}{\ensuremath{Q}}
\newcommand{\pos}[1]{\ensuremath{\mathsf{p}}(#1)}
\newcommand{\ngt}[1]{\ensuremath{\mathsf{n}}(#1)}
\newcommand{\tsconjp}[1]{\ensuremath{\mathsf{P}\textsf{-}\mathsf{tsconj}}(#1)}
\newcommand{\faspp}[1]{\ensuremath{\mathsf{P}\textsf{-}\mathsf{fASP}}(#1)}
\newcommand{\aux}[1]{\ensuremath{\mathsf{aux}_{#1}}}

\newcommand{\nnf}[1]{\ensuremath{\mathsf{NNF}}(#1)}

\newcommand{\toaspalgname}[1]{\ensuremath{\mathsf{ToASP}(#1)}}
\newcommand{\projaspcount}[1]{\ensuremath{\#\mathsf{PASP}(#1)}}

\usepackage{color}

\usepackage{algorithm}
\usepackage{algpseudocode}
\algnewcommand\algorithmicforeach{\textbf{for each}}
\algdef{S}[FOR]{ForEach}[1]{\algorithmicforeach\ #1\ \algorithmicdo}

%% file: sections/introduction.tex
\section{Introduction}
\label{sec:introduction}

Boolean Networks (BNs) serve as a fundamental modeling framework for representing complex dynamical systems across domains such as systems biology, computational logic, and artificial intelligence~\cite{DWM2024,HMJ2024,Rozum2021,SKIKK2020,TBSF2024}.
Their ability to capture intricate interactions and complex system behaviors makes them a valuable tool for studying diverse phenomena such as gene regulatory networks, logical circuits, and reasoning processes~\cite{ASDO2024,Rozum2021}.
A crucial aspect of BN dynamics is the existence of trap spaces—subspaces of the state space that, once entered, cannot be exited~\cite{HAH2015}.
These structures are instrumental in understanding the long-term behavior of BNs, as they often correspond to stable system configurations or attractors~\cite{HAH2015}.

Among trap spaces, minimal trap spaces (fixed points are a special case of minimal trap spaces in which all variables are {\em fixed}~\cite{HAH2015}) are of particular interest due to their role in characterizing the dynamical landscape of a BN~\cite{HAH2015,TPPR2024}.
The identification and counting of minimal trap spaces provide valuable insights into the system stability, attractor structure, and probabilistic behavior~\cite{HAH2015,CS2015}.
However, enumerating these structures poses a significant challenge due to their computational complexity, especially for large-scale BNs encountered in practical applications~\cite{TBPS2024}.
Scalable methods for counting minimal trap spaces are therefore essential for advancing the BN analysis.

Unfortunately, minimal trap spaces have received little attention in the context of counting problems for Boolean networks.
In contrast, the problem of counting fixed points has been studied in several works~\cite{PG2005,Predrag2006,CS2015,FAKA2022}.
These works show that fixed point counting is generally \#P-complete and examine its complexity under structural restrictions~\cite{PG2005,Predrag2006} or specific classes of Boolean update functions~\cite{CS2015}.
Some studies also consider scenarios in which the update logic is only partially known~\cite{FAKA2022}.
On the practical side, while no dedicated implementation exists for fixed point counting, the task can be reduced to propositional model counting~\cite{SRSM19} via a CNF encoding of fixed points~\cite{EM2011}.
To the best of our knowledge, no prior work---either theoretical or practical---has addressed the problem of counting minimal trap spaces in Boolean networks.

Our paper addresses the research gap by formulating six meaningful problems related to counting minimal trap spaces and fixed points in BNs.
These counting problems capture core tasks such as counting all minimal trap spaces or fixed points, counting those that satisfy a given {\em property} (e.g., a {\em phenotype}), and counting solutions under {\em perturbations}.
These problems provide valuable insights not only within BN theory (e.g., probabilistic reasoning and dynamical analysis) but also in broader application areas such as abstract argumentation and logic programming (discussed in~\Cref{sec:problem-formulation}).
We subsequently propose novel and efficient methods to solve the counting problems by exploiting the expressive power of Answer Set Programming (ASP)~\cite{MT1999}.
ASP is a declarative problem solving paradigm and has widely been applied in the field of systems biology~\cite{ST2009,VGETGNSSS2015}, in particular in the analysis and control of BNs~\cite{AFRM2017,KSSV2013,HAH2015,Paulev2020,TBPS2024,TBS2023} (see~\Cref{sec:related-work}).
As in existing ASP counting literature~\cite{KM2023}, our ASP-based reduction leverages the approximate answer set counter ApproxASP~\cite{KESHFM22}, which employs a {\em hashing-based} technique for approximate answer set counting.
Finally, we conduct an extensive experimental evaluation on a diverse benchmark dataset.
Our analysis shows that ApproxASP efficiently estimates the number of minimal trap spaces and fixed points, and ApproxMC~\cite{YM2023} efficiently estimates fixed point counts. 
Both approaches avoid exhaustive enumeration, significantly improving the feasibility of counting compared to enumeration and BDD-based methods used in other tools.

The remainder of the paper is structured as follows. 
In~\Cref{sec:preliminaries}, we review the necessary background on propositional logic, BNs, answer set programming, and model counting.
\Cref{sec:related-work} surveys related work.
\Cref{sec:problem-formulation} defines the six counting problems we consider and discusses their applications.
\Cref{sec:computation-methods} introduces our ASP-based reduction methods for solving these problems. 
\Cref{sec:experimental-results} presents experimental results, demonstrating the effectiveness of our methods. 
Finally, \Cref{sec:conclusion} concludes the paper and outlines potential directions for future research.

%% file: sections/preliminaries.tex
\section{Preliminaries}\label{sec:preliminaries}
In this section, we present background and some notations from propositional logic, Boolean networks, and answer set programming.

\subsection{Propositional Formulas}

In this work, we employ \(\twod = \{0, 1\}\) as the Boolean domain and \(\threed = \{0, 1, \star\}\) as the three-valued domain.
Under propositional semantics, each {\em propositional variable} $v$ takes its value from $\twod$. 
A propositional formula is defined recursively: the basic formulas include the Boolean constants $0$ and $1$, any variable $v$, and its negation $\neg{v}$. 
Moreover, if a formula is prefixed by $\neg$ (not) or if two formulas are connected by a logical connective --- \(\land\) (conjunction), \(\lor\) (disjunction), and \(\leftrightarrow\) (bi-implication) --- the resulting expression is also a formula.
A formula is in Conjunctive Normal Form (CNF) if it is expressed as a conjunction of disjunctions of literals (variables or their negations), and it is in Negation Normal Form (NNF) if all negations are applied directly to propositional variables.
Any propositional formula can be converted into a semantically equivalent CNF or NNF by recursively applying specific rewriting rules for its logical connectives~\cite{BHV2009,MSS2023}.

\subsection{Boolean Networks}\label{subsec:pre-bn}

A Boolean Network (BN) \(\bn\) is defined as a finite set of Boolean functions over a finite set of Boolean variables, denoted by \(\var{\bn}\).
Each variable \(v \in \var{\bn}\) is associated with a Boolean function \(f_v \colon \twod^{|\var{\bn}|} \to \twod\).
A function \(f_v\) is termed \emph{constant} if it is always either 0 or 1 regardless of the values of its arguments.
A variable \(v\) is considered a \emph{source variable} if \(f_v\) is the identity function on \(v\), i.e., \(f_v = v\).
A state \(s\) of \(f\) is a Boolean vector \(s \in \twod^{|\var{\bn}|}\) that can be viewed as a mapping: \(s \colon \var{\bn} \to \twod\);
we denote the value of variable $v$ in state $s$ by \(s_v\).
For convenience, a state is often represented as a string of values (e.g., ``0110'' instead of (0, 1, 1, 0)).

At each discrete time step \(t\), each variable \(v\) can update its state  according to its Boolean function $f_v$; that is, $v$'s state at time $t+1$ is given by \(s'_v = f_v(s)\).
An \emph{update scheme} specifies how these state updates occur over time~\cite{SKIKK2020}.
The two primary schemes are synchronous, in which all variables update simultaneously, and fully asynchronous, where a single variable is chosen non-deterministically to update.
Under arbitrary update scheme, the BN transitions from one state to another --- a process known as a {\em state transition}. The overall dynamics of the BN are captured by the {\em State Transition Graph} (STG), a directed graph whose nodes represent states and edges represent transitions.
We denote the STG under the synchronous update scheme as  \(\stg{\bn}\) and that under the fully asynchronous scheme as \(\atg{\bn}\).

A non-empty set \(A\) of states is a \emph{trap set} if there is no transition from a state in $A$ to a state outside $A$ in the State Transition Graph (STG) of $\bn$ (i.e., there is no pair $x \in A$ and $y \not \in A$ such that \((x, y)\) is an arc in the STG)~\cite{HAH2015}.
A trap set that is minimal with respect to set inclusion is termed an {\em attractor}.
In particular, an attractor containing a single state is called a {\em fixed point}, while one with two or more states is referred to as a {\em cyclic} attractor.
A \emph{sub-space} \(m\) of a BN \(\bn\) is a mapping \(m \colon \var{\bn} \to \threed\).
A variable \(v \in \var{\bn}\) is said to be \emph{fixed} (resp.\ \emph{free}) in \(m\) if \(m(v) \neq \star\) (resp.\ \(m(v) = \star\)).
For convenience, a sub-space is often represented as a string of values  (e.g., \(0\star\) instead of \(\{v_1 = 0, v_2 = \star\}\)).
The sub-space \(m\) represents a set of states, denoted by \(\mathcal{S}[m]\), defined as 
\[
\mathcal{S}[m] = \{s \in \mathbb{B}^{|\var{\bn}|} \mid s_v = m(v), \forall v \in \var{\bn}, m(v) \neq \star\}
\]
For example, if \(m = \star11\), then \(\mathcal{S}[m] = \{011, 111\}\).
If a sub-space is also a trap set, it is a \emph{trap space}.
Unlike trap sets and attractors, trap spaces are independent of the update scheme employed~\cite{HAH2015}.
Notably, a fixed point of \(\bn\) is a special trap space in which all variables are fixed.
A trap space \(m\) is \emph{minimal} if there is no trap space \(m'\) such that \(\mathcal{S}[m'] \subset \mathcal{S}[m]\).
Since an attractor is a subset-minimal trap set, a minimal trap space contains at least one attractor of the BN, regardless of the update scheme employed~\cite{HAH2015}.

\begin{example}\label{exam:straight-BN}
	Let us consider BN \(\bn\) with \(\var{\bn} = \{a, b\}\), \(f_a = a \land \neg b\), and \(f_b = a\).
	The synchronous STG of \(\bn\) is shown in Figure~\ref{fig:sstg-straight-BN-wildtype-knockout}a.
	The set \(\{00, 01, 11\}\) is a trap set but not a trap space.
	It is easy to check that \(\bn\) has three trap spaces: \(m_1 = 00\), \(m_2 = 0\star\), and \(m_3 = \star\star\).
	Among these, \(m_1\) is a minimal trap space (also a fixed point) of \(\bn\). 
	In this case, $m_1$ is also the only synchronous attractor of $\bn$.
\end{example}

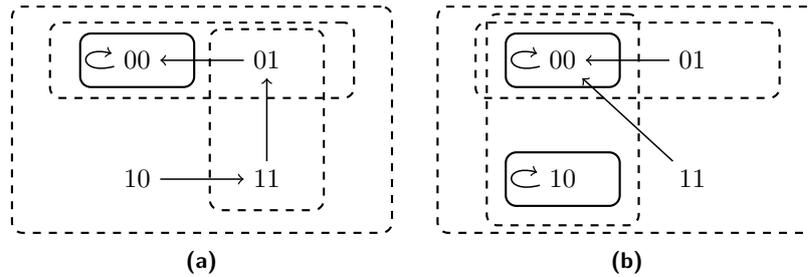
\begin{figure}[!ht]
	\centering
	\captionsetup[subfigure]{justification=centering}
	\begin{subfigure}{0.4\textwidth}
		\centering
		\begin{tikzpicture}[node distance=0.8cm and 0.8cm, every node/.style={scale=1.0}, line width = 0.2mm]
			\node[] (0) [] {00};
			\node[] (1) [right=of 0,xshift=0.3cm] {01};
			\node[] (2) [below=of 0,yshift=-0.3cm] {10};
			\node[] (3) [right=of 2,xshift=0.3cm] {11};
			
			\draw[->] (0) edge [loop left] (0);
			\draw[->] (1) edge [] (0);
			\draw[->] (2) edge [] (3);
			\draw[->] (3) edge [] (1);
			
			\node[draw=black, thick, rounded corners, fit=(0), minimum width=1.5cm, minimum height=0.5cm] {};
			\node[draw=black, dashed, thick, rounded corners, fit=(0)(1), minimum width=4cm, minimum height=1cm] {};
			
			\node[draw=black, dashed, thick, rounded corners, fit=(1)(3), minimum width=1.5cm, minimum height=2.4cm] {};
			\node[draw=black, dashed, thick, rounded corners, fit=(0)(1)(2)(3), minimum width=5cm, minimum height=3cm] {};
		\end{tikzpicture}
		\caption{}
	\end{subfigure}\begin{subfigure}{0.4\textwidth}
		\centering
		\begin{tikzpicture}[node distance=0.8cm and 0.8cm, every node/.style={scale=1.0}, line width = 0.2mm]
			\node[] (0) [] {00};
			\node[] (1) [right=of 0,xshift=0.3cm] {01};
			\node[] (2) [below=of 0,yshift=-0.3cm] {10};
			\node[] (3) [right=of 2,xshift=0.3cm] {11};
			
			\draw[->] (0) edge [loop left] (0);
			\draw[->] (1) edge [] (0);
			\draw[->] (2) edge [loop left] (2);
			\draw[->] (3) edge [] (0);
			
			\node[draw=black, thick, rounded corners, fit=(0), minimum width=1.5cm, minimum height=0.5cm] {};
			\node[draw=black, dashed, thick, rounded corners, fit=(0)(1), minimum width=4cm, minimum height=1cm] {};
			\node[draw=black, thick, rounded corners, fit=(2), minimum width=1.5cm, minimum height=0.5cm] {};
			\node[draw=black, dashed, thick, rounded corners, fit=(0)(2), minimum width=2cm, minimum height=2.8cm] {};
			\node[draw=black, dashed, thick, rounded corners, fit=(0)(1)(2)(3), minimum width=5cm, minimum height=3cm] {};
		\end{tikzpicture}
		\caption{}
	\end{subfigure}	
	\caption{(a) Synchronous STG $\stg{\bn}$ of BN $\bn$ from Example~\ref{exam:straight-BN}. (b)~Synchronous STG \(\stg{\bn}\) where variable $b$ is subject to a {\em knockout} (i.e., its value is forced to $0$). Trap spaces (resp.\ minimal trap spaces) are enclosed by dashed (resp.\ solid) rectangular frames.}%
	\label{fig:sstg-straight-BN-wildtype-knockout}
\end{figure}

\subsection{Answer Set Programming}

An Answer Set Program (ASP) $P$ expresses logical constraints
between a set of propositional variables.
In answer set programming, such variables are also called \emph{atoms}, and the
set of atoms appearing in $P$ is denoted by $\at{P}$.  
For notational convenience, we will use the terms ``variable'' and ``atom'' interchangeably throughout the paper.

\noindent An ASP is a finite set of \emph{rules} of the form 
 \[
 a_1 \lor \dots \lor a_k \leftarrow b_1, \dots, b_m, \dng{c_1}, \dots, \dng{c_n}
 \]
 where \(k, m, n \geq 0\) and $\dng{}$ denotes \textit{default negation}~\cite{clark1978}.
 Given a rule \(r\) of the above form, we define \(\head{r} = \{a_1, \dots, a_n\}\), \(\pbody{r} = \{b_1, \dots, b_m\}\), and \(\nbody{r} = \{c_1, \dots, c_n\}\) as the \emph{head}, \emph{positive body}, and \emph{negative body} of \(r\), respectively.
 The rule \(r\) is a \emph{fact} if \(\pbody{r} = \nbody{r} = \emptyset\) and an (integrity) \emph{constraint} if \(\head{r} = \emptyset\).
 We often use the notations $\top$ and $\bot$ to denote empty body ($\pbody{r} \cup \nbody{r} = \emptyset$) and empty head ($\head{r} = \emptyset$), respectively.
 The rule \(r\) is called \emph{normal} if \(|\head{r}| \leq 1\), and the ASP program is normal (or a normal logic program) if all its rules are normal.
 A program is called {\em disjunctive} if there is a rule $r \in P$ such that $|\head{r}| > 1$~\cite{BD1994}.

\paragraph{Answer Set Semantics.} A subset \(M\) of atoms (called an \emph{interpretation}) \emph{satisfies} a rule \(r\) if \((\head{r} \cup \nbody{r}) \cap M \neq \emptyset\) or \(\pbody{r} \setminus M \neq \emptyset\).
The interpretation \(M\) is a \emph{supported model} of \(P\) if it satisfies all rules of \(P\), denoted as \(M \models P\).
The \emph{Gelfond-Lifschitz (GL) reduct} of the program \(P\) with respect to the interpretation \(M\) is defined as follows: \(P^{M} := \{\head{r} \leftarrow \pbody{r} | r \in P, M \cap \nbody{r} = \emptyset\}\)~\cite{gelfond1988stable}.
The interpretation \(M\) is an \emph{answer set} (or a \emph{stable model}) of \(P\) if $\not \exists M\textprime \subsetneq M$ such that $M\textprime \models P^M$.
We use the notation \(\as{P}\) to denote the set of answer sets of \(P\).

\paragraph{Answer Set Counting.}\label{subsec:pre-answer-set-counting}
Given an ASP \(P\), the {\em answer set counting} problem (denoted as \#ASP) seeks to compute the number of answer sets of \(P\), written as $\Card{\as{P}}$.
The projected answer set counting problem (denoted as \#PASP) extends \#ASP by counting the number of distinct answer sets of $P$ with respect to a given set of \emph{projection atoms} \(I \subseteq \at{P}\).
Two answer sets are considered equivalent if they differ only on atoms in $\at{P} \setminus I$~\cite{FH2019}.
We denote the set of projected answer sets as $\as{P, I}$, so that \#PASP computes \(|\{M \cap I \mid M \in \as{P}\}|\).
As a special case, when \(I = \at{P}\), \#PASP reduces to \#ASP.
We use the notation $\projaspcount{P, I}$ to denote the {\em projected answer set count} of program $P$ w.r.t. projection atoms $I$.

In the probably approximately correct (PAC) framework for answer set counting, given a program $P$ and a projection set $I$, the goal is to estimate a count $\cnt$ satisfying
\begin{equation*}
	\text{Pr}\left[\frac{\Card{\as{P, I}}}{(1 + \varepsilon)} \leq \cnt \leq (1 + \varepsilon) \times \Card{\as{P, I}}\right] \geq 1 - \delta
\end{equation*}
where $0 < \varepsilon < 1$ is the {\em tolerance} and $0 < \delta < 1$ is the {\em confidence}~\cite{KESHFM22,CMV2016}.

%% file: sections/related-work.tex
\section{Related Work}\label{sec:related-work}

\paragraph{ASP-Based Computation of Fixed Points and Minimal Trap Spaces.}
Several ASP encodings have been proposed to characterize fixed points and minimal trap spaces in BNs.
In many cases, an ASP encoding for fixed points is derived from its minimal trap space counterpart by adding {\em integrity constraints} to capture the specific properties of fixed points~\cite{VML2024,Trinh2023}.
The first ASP encoding that requires the computation of {\em prime implicants} for each Boolean function was proposed and implemented in~\cite{HAH2015}.
A major bottleneck of this encoding is that computing even a single prime implicant is NP-hard and the total number of prime implicants can be exponential in the number of function inputs~\cite{CM78}.
Subsequent encodings~\cite{Paulev2020,VML2024,Trinh2023} that were proposed to overcome this bottleneck still suffer from scalability and efficiency issues, particularly for very large and complex models, primarily because they require the {\em disjunctive normal forms} of all the Boolean functions of the original BN.
For fixed points, the ASP encoding by Trinh et al.~\cite{TBS2023} (called \fasp) uses a {\em negation normal form} for each Boolean function, which is much more efficient to obtain.
This encoding was later generalized for minimal trap spaces~\cite{TBPS2024} (called \tsconj); however, when dealing with {\em unsafe formulas} that might yield incorrect solutions, it still requires a disjunctive normal form to ensure the correctness.

\paragraph{BN Encoding of Normal Logic Programs.} The theoretical work by Inoue~\cite{Inoue2011} was among the first to establish a connection between ASP and BNs.
It defines a BN encoding for finite ground normal logic programs, which relies on the notion of the Clark’s completion~\cite{clark1978}, and points out that the two-valued models of the Clark's completion of a finite ground normal logic program one-to-one correspond to the fixed points of the encoded BN.
The subsequent work~\cite{IS2012} points out that the strict supported classes of a finite ground normal logic program one-to-one correspond to the synchronous attractors of the encoded BN.
Very recently, Trinh et al.~\cite{TBSF2024} related the regular models in a finite ground normal logic program and the minimal trap spaces in the respective BN, and further applied these theoretical results to explore graphical conditions for the existence, uniqueness, and number of regular models.

\paragraph{SAT Characterization of Fixed Points.} The set of fixed points of a BN \(\bn\) can be characterized as the set of {\em satisfying assignments} of the propositional formula \(\bigwedge_{v \in \var{\bn}}\left(v \leftrightarrow f_v\right )\)~\cite{EM2011}.
Hence, we can apply \#SAT tools~\cite{Thurley2006,CSV2013,SRSM19} to counting the number of fixed points of a BN.
To the best of our knowledge, to date, there is no SAT characterization for the set of minimal trap spaces of a BN.

\paragraph{Answer Set Counting.} For general ASP programs, \#ASP is \(\#\cdot\text{coNP-complete}\)~\cite{FHMW2017},
while \#ASP is \(\#\cdot\text{P-complete}\) for normal ASP programs, which follows from standard reductions~\cite{Janhunen2006}.
The projected answer set counting \#PASP is simply \(\Sigma_2^p\)-complete if the set of projection atoms is empty; otherwise, the complexity is \(\#\cdot\Sigma_2^p\)-complete~\cite{FH2019}.
Fichte et al.~\cite{FH2019,FHMW2017} exploited the {\em tree decomposition}-based technique for counting answer sets.
Kabir et al.~\cite{KESHFM22} introduced the hashing-based approximate counting technique for answer set counting. %

%% file: sections/problem_formulate.tex
\section{Problem Formulation}\label{sec:problem-formulation}

In this section, we introduce several counting problems related to minimal trap spaces in Boolean networks (BNs), covering both minimal trap spaces and fixed points.
We begin with a straightforward counting variant, 
then propose a specialized variant that requires solutions to satisfy a specific {\em property}, 
and finally present a more complex variant focused on phenotype measurement in BNs under perturbations. To highlight the biological utility of these theoretical problems, a brief case study is also given in Appendix~\ref{sec:case-study}.

\subsection{Counting Minimal Trap Spaces and Fixed Points}\label{subsec:formulation-first}

We introduce two fundamental counting problems for Boolean networks (BNs): one that counts the number of minimal trap spaces (\Cref{def:problem-first-MTS}) and another one counts the number of fixed points (\Cref{def:problem-first-FIX}). 
These problems address the basic question: How many minimal trap spaces or fixed points does a given BN have?

\begin{definition}[\acfirstmts{}]\label{def:problem-first-MTS}
	Given a BN \(\bn\), compute the number of minimal trap spaces of \(\bn\).
\end{definition}

\begin{definition}[\acfirstfix{}]\label{def:problem-first-FIX}
	Given a BN $\bn$, compute the number of fixed points of $\bn$.
\end{definition}

In BN research, both counting problems --- \acfirstmts{} and \acfirstfix{} --- are valuable when full enumeration is infeasible, as in gene regulatory network models with many source variables~\cite{SVATSAJ2020,TBPS2024,TBS2023}.
They are also useful when divergent solutions are sought~\cite{SVLAL2020}. 
Notably, the fixed point counting problem (\acfirstfix{}) can enhance methods for enumerating asynchronous attractors by enabling the selection of a smaller candidate set, thereby potentially speeding up the filtering process~\cite{TKB2022}.
Finally, both \acfirstmts{} and \acfirstfix{} can lay the groundwork for probabilistic reasoning in BNs~\cite{SDZ2002}.

Recent research has connected Boolean networks (BNs) with two broad fields: abstract argumentation and logic programming.
Although we omit the preliminaries on Abstract Argumentation Frameworks (AFs)~\cite{DFGH2022}, Abstract Dialectical Frameworks (ADFs)~\cite{LMNWW22}, and Normal Logic Programs (NLPs)~\cite{TBSF2024}, interested readers can consult the cited papers for details.
Specifically, Trinh et al.~\cite{TBV2025} and Dimopoulos et al.~\cite{DWM2024}  demonstrate that the {\em preferred} (resp.\ {\em stable}) extensions of an AF one-to-one correspond to the minimal trap spaces (resp.\ fixed points) of the respective BN.
Similarly,  Heyninck et al. \cite{HMJ2024} and Azpeitia et al.~\cite{ASDO2024} demonstrate that ADFs and BNs are identical --- in this context, the {\em preferred} (resp.\ {\em two-valued}) interpretations of an ADF match the minimal trap spaces (resp.\ fixed points) of the respective BN.
In the realm of logic programming, the subset-minimal {\em supported partial} models (resp.\ {\em supported} models) of a {\em finite ground} NLP correspond one-to-one with the minimal trap spaces (resp.\ fixed points) of the respective BN~\cite{Inoue2011,TBSF2024}.
Moreover, if the finite ground NLP is {\em tight}~\cite{Fages1994,LL03}, then its regular models (resp.\ stable models) one-to-one correspond to the minimal trap spaces (resp.\ fixed points) of the respective BN.
These types of extensions, interpretations, and models are central to the study of AFs, ADFs, and finite ground NLPs~\cite{FHM2024,JNSSY2006,LMNWW22}.

Recent work on AFs has focused on counting stable and preferred extensions, yielding both complexity results~\cite{FHM2024} and dynamic programming methods~\cite{DFGH2022} (note that the dynamic programming methods do not support preferred extensions).
In contrast, counting problems for ADFs and finite ground NLPs remain largely unexplored, aside from a few studies on answer sets~\cite{DBLP:conf/aaai/AzizCMS15,KESHFM22}.
Thanks to the connections between BNs and these formalisms, these results from \acfirstmts{} and \acfirstfix{} can be extended to ADFs --- where they correspond to preferred and two-valued interpretations --- as well as to general finite ground NLPs (for supported partial and supported models) and tight finite ground NLPs (for regular and stable models).

\subsection{Counting with Satisfying Properties}\label{subsec:formulation-second}

We examine a specialized variant of problems \acfirstmts{} and \acfirstfix{}, focusing on counting minimal trap spaces (Definition~\ref{def:problem-second-MTS}) and fixed points (Definition~\ref{def:problem-second-FIX}) that satisfy a specified property. 
This formulation addresses the natural question: How many minimal trap spaces (or fixed points) in a BN exhibit a given property or assumption?

In systems biology, BNs are used to model biological phenotypes, which reflect an organism's functional characteristics~\cite{LQADZXHE2016}.
Several definitions of phenotype in BNs have been proposed~\cite{BBPSS2023,KHNS2018,KTFJCS2021}; in this work, we adopt one of the most widely used notions.
Given a BN $f$, we define {\em trait} as a statement of the form $(v \leftrightarrow e)$, where $v \in \var{f}$ and $e \in \threed$.
Note that \(v \leftrightarrow \star\) is evaluated true if and ony if \(v = \star\).
A phenotype $\phen$ is then defined as the conjunction of a set of traits.
A sub-space $m$ satisfies a phenotype $\phen$ (denoted by \(m \models \phen\)) if, upon replacing each variable $v \in \phen$ with its value $m(v)$, the resulting formula evaluates to true under propositional semantics.
Unlike minimal trap spaces, fixed points require that all variables take Boolean values ($e \in \twod$), and hence phenotypes involving `$\star$' are not meaningful in this context.

In our problem formulation, the property of interest is a desirable phenotype.
In systems biology, a minimal trap space satisfying a phenotype suggests the phenotype's potential emergence in vivo.
Furthermore, this counting variant is applicable beyond systems biology. 
For instance, in abstract argumentation frameworks (resp.\ normal logic programs), a phenotype may represent the presence of particular set of arguments (resp.\ atoms) in an extension (resp.\ a model)~\cite{FHM2024,JNSSY2006}, a concept closely linked to {\em credulous reasoning}.

\begin{definition}[\acsecondmts{}]\label{def:problem-second-MTS}
	Given a BN \(\bn\) and a phenotype $\phen$, compute the number of minimal trap spaces of \(\bn\) that satisfy $\phen$.
\end{definition}

\begin{definition}[\acsecondfix{}]\label{def:problem-second-FIX}
	Given a BN \(\bn\) and a phenotype $\phen$, compute the number of fixed points of \(\bn\) that satisfy $\phen$.
\end{definition}

When the specified phenotype is a {\em tautology}, \acsecondmts{} (resp. \acsecondfix{}) reduces directly to \acfirstmts{} (resp. \acfirstfix{}).
Thus, all implications discussed for \acfirstmts{} and \acfirstfix{} also apply to \acsecondmts{} and \acsecondfix{}, respectively.
In the following, we examine these applications through the lens of abstract argumentation frameworks.

First, \acsecondmts{} and \acsecondfix{} facilitate more nuanced reasoning between {\em skeptical} and {\em credulous} approaches by integrating quantitative and probabilistic methods~\cite{FHN2022}. 
For instance, by running \acsecondmts{} (or \acsecondfix{}) twice --- once with a specific set of arguments and once with a tautology --- we can gain insights into preferred (or stable) extensions~\cite{FHM2024}.
Second, this approach shifts reasoning from simple decision-making (i.e., determining whether a set of arguments is present in an extension) to probabilistic reasoning (i.e., assessing the likelihood of a set of arguments appears in an extension). 
Moreover, this framework supports the development of advanced probabilistic semantics for abstract argumentation frameworks~\cite{KBDDGH2022}.
To our knowledge, these aspects have not yet been explored in the context of BNs, and our work paves the way for integrating such reasoning into BN analysis.

\begin{example}\label{exam:four-first-problems}
	Consider again the BN \(\bn\) shown in~\Cref{exam:straight-BN}.
	It has a unique minimal trap space \(m_1 = 00\), which is also its only fixed point of \(\bn\).
	Hence, the answer to \acfirstmts{} (resp.\ \acfirstfix{}) is $1$.
	Considering the phenotype \(\phen = (b \leftrightarrow \star)\), the answer to \acsecondmts{} is $0$.
\end{example}

\subsection{Counting Under Perturbations and Measuring Robustness}\label{subsec:formulation-third}
We consider a main contribution of this paper to be the identification of new, relevant counting problems (other than \acsecondmts{} and \acsecondfix{}, which were previously known): \acthirdmts{} and \acthirdfix{}.
To formalize these problems, we introduce the concept of a {\em perturbation}.

Consider a BN \(\bn\).
A \emph{perturbation} $\sigma$~\cite{su2020sequential} on a set \(\pert \subseteq \var{f}\) of {\em perturbable} variables is defined as a mapping from \(\pert \to \mathbb{B}_{\star}\).
In practice, $\pert$ can be any subset of variables whose perturbation is biologically meaningful.
Since each variable in $\pert$ can assume one of three values under perturbation, there are \(3^{|\pert|}\) possible perturbations.
For each \(v \in \pert\), setting \(\sigma(v) = 0\) forces \(f_v = 0\) ({\em knockout perturbation}), setting \(\sigma(v) = 1\) forces \(f_v = 1\) ({\em over-expression perturbation}), and setting \(\sigma(v) = \star\) leaves \(f_v\) unchanged.
Consequently, the perturbed BN, denoted \(\bn^{\sigma}\), is defined by \(\var{\bn^{\sigma}} = \var{\bn}\) and, for every \(v \in \var{\bn^{\sigma}}\), 

$$
f^{\sigma}_v =\begin{cases}
		 \sigma(v) & \text{if \(v \in \pert\) and \(\sigma(v) \neq \star\)}\\
         f_v & \text{otherwise}
		 \end{cases}
$$Note that, the value \(\star\) is used to distinguish imperturbable variables and perturbable variables that are unchanged under a certain perturbation.

In biological systems, a perturbation refers to any disturbance that disrupts the normal functioning of a BN. 
Such disturbances may arise from {\em genetic mutations}~\cite{shmulevich2002gene}, external factors such as {\em medications}~\cite{bloomingdale2018boolean}, or other influences~\cite{montagud2022patient}. 
These perturbations can substantially alter the phenotypes exhibited by a BN, making it crucial to quantify their effects on network behavior. 
In systems biology, this impact is commonly assessed in terms of {\em robustness}~\cite{kitano2007towards}, which motivates our definitions of problems \acthirdmts{} and \acthirdfix{}.

\begin{definition}[\acthirdmts{}]\label{def:problem-third-MTS}
	Given a BN \(\bn\), a set of perturbable variables \(\pert\), and a target phenotype \(\phen\),
	determine the number of perturbations \(\sigma\) on \(\pert\) such that the perturbed BN \(\bn^{\sigma}\) exhibits at least one minimal trap space that satisfies \(\phen\).
\end{definition}

\begin{definition}[\acthirdfix{}]\label{def:problem-third-FIX}
	Given a BN \(\bn\), a set of perturbable variables \(\pert\), and a target phenotype \(\phen\),
	determine the number of perturbations \(\sigma\) on \(\pert\) such that the perturbed BN \(\bn^{\sigma}\) exhibits at least one fixed point that satisfies \(\phen\).
\end{definition}

Leveraging the results of \acthirdmts{} (or \acthirdfix{} if focusing solely on fixed points), we define the robustness of a phenotype as the fraction of perturbations that preserve the phenotype relative to the total number of perturbations ($3^{\Card{\pert}}$). 
In other words, robustness measures the probability that a phenotype remains active following a random, admissible perturbation is applied to the network. 
This measure of phenotype robustness can inform the selection of specific perturbations and guide targeted treatment strategies~\cite{BBPSS2023,TP2024}. A small case study on an interferon model with 121 variables presents these concepts more practically in Appendix~\ref{sec:case-study}.

\begin{example}\label{exam:third-MTS}
	Consider BN \(\bn\) given in~\Cref{exam:straight-BN}.
	Consider $\pert = \{b\}$ denote the set of perturbable variables and define the set of desirable phenotype as \(\phen = (a \leftrightarrow 0 \land b \leftrightarrow 0)\).
	There are three possible perturbations: \(\sigma_1 = \{b = 0\}\), \(\sigma_2 = \{b = 1\}\), and \(\sigma_3 = \{b = \star\}\).
	The perturbation \(\sigma_1\) represents the knockout perturbation of variable \(b\) and \(\stg{f^{\sigma_1}}\) is given in~\Cref{fig:sstg-straight-BN-wildtype-knockout}b.
	It is straightforward to observe that \(\bn^{\sigma_1}\) has two minimal trap spaces $00$ and $10$.
	In constrast, the BN \(\bn^{\sigma_2}\) has one minimal trap space $01$, which does not satisfy the given phenotype.
	The BN \(\bn^{\sigma_3}\) equals \(\bn\) and has one minimal trap space $00$.
	Therefore, for BN $f$, phenotype $\phen$ and perturbable variables $\pert$, the answer to $\acthirdmts{}$ is $2$ 
	and the perturbation robustness of phenotype $\phen$ in BN \(\bn\) w.r.t. $\pert$ is \(\sfrac{2}{3}\).
\end{example}

\noindent\textbf{On the impact of precision} Finally, it should be noted that while these problems are defined \emph{exactly}, in practice their results mainly serve as a means of comparison. For example, the results of C-MTS-2 could be used to compare the abundance of two biological phenotypes. Then the exact count may not be important, as long as the two phenotypes can be compared reliably. As such, these problems are particularly suitable for approximate counting. This also impacts the choice of method parameters ($\epsilon$ and $\delta$), as in practice, even low precision (as used in our benchmarks) can be sufficient to distinguish between significantly different phenotypes. For closely matched results, the precision can be then increased as needed.

%% file: sections/computation.tex
\section{Computational Methods}\label{sec:computation-methods}

Given a BN \(\bn\), our approach is to construct an ASP program \(P\) such that the answer sets of \(P\) one-to-one correspond to the minimal trap spaces (or fixed points) of \(\bn\).
This reduction allows us to leverage existing answer set counters to efficiently count the answer sets of $P$.
For \acfirstmts{} and \acfirstfix{}, we simply use the ASP encodings of \tsconj and \fasp, respectively.
For \acsecondmts{} and \acsecondfix{}, we complement the encoding of the given phenotype.
For \acthirdmts{} and \acthirdfix{}, we propose a new perturbation encoding, which we consider to be a main contribution of this paper.
Now, we begin by briefly reviewing the \tsconj and \fasp encodings.

\subsection{\tsconj and \fasp Encodings}

We first present the common components of the two encodings, followed by their differences.
The ASP encodings represent sub-spaces using atoms $\pos{v}$ and $\ngt{v}$, indicating whether variable $v$ is fixed to $1$, $0$, or left free. 
The goal is to translate BN trap space and fixed point properties into ASP rules whose answer sets correspond to these configurations.

Given a BN \(f\), the encodings compute an ASP program \(P\) as follows:
for each variable \(v \in \var{f}\), two atoms \(\pos{v}\) and \(\ngt{v}\) are introduced to indicate positive and negative assignments of the variable $v$, respectively.
Additionally, for every \(v \in \var{f}\), one rule of the form: \(\pos{v} \vee \ngt{v} \leftarrow \top\) is added to ensure that each answer set corresponds to a sub-space of $f$.
The translation from an answer set $M$ to a sub-space $m$ is defined as follows: for each \(v \in \var{f}\), we have (i) \(m(v) = 1\) if and only if \(\pos{v} \in M \land \ngt{v} \not \in M\), (ii) \(m(v) = 0\) if and only if \(\pos{v} \not \in M \land \ngt{v} \in M\), and (iii) \(m(v) = \star\) if and only if \(\pos{v} \in M \land \ngt{v} \in M\).
Recall that a trap space of \(f\) can be characterized by \(\bigwedge_{v \in \var{f}}(v \leftarrow f_v) \land (\neg v \leftarrow \neg f_v)\)~\cite{TBPS2024}.
For every \(v \in \var{f}\), two ASP rules --- (i) \(\gamma(v) \leftarrow \gamma(\nnf{f_v})\) and (ii) \(\gamma(\neg v) \leftarrow \gamma(\nnf{\neg f_v})\) --- are added to $P$ to express the characterization, where \(\nnf{\Phi}\) denotes a negation normal form of a Boolean formula \(\Phi\) and \(\gamma\) is a procedure defined as follows:
\begin{align*}
	\gamma(v) &= \pos{v}, \gamma(\neg v) = \ngt{v}, v \in \var{f}, \\
	\gamma(\bigwedge_{1\leq j\leq J}\alpha_{j}) &= \gamma(\alpha_{1}) \land \ldots \land \gamma(\alpha_{J}), \gamma(\bigvee_{1\leq j\leq J}\alpha_{j}) = \aux{k},
\end{align*}where \(\aux{k}\) is a new \emph{auxiliary} atom, \(k\) is a global counter starting from $1$ and shall be increased by $1$ after each new auxiliary atom is created, and for each \(j\), the rule \(\aux{k} \gets \gamma(\alpha_{j})\) is added to \(P\).

For the correctness of the \tsconj encoding, Trinh et al.~\cite{TBPS2024} defined a syntactic safeness condition for a Boolean formula $\Phi$: \(\Phi\) is considered safe if it does not contain any conjunction of two subformulas \(\Phi_1\) and \(\Phi_2\) such that there exists a variable \(x\) appearing in \(\Phi_1\) with \(\neg x\) appearing in \(\Phi_2\).
When both \(f_v\) and \(\neg f_v\) are safe for every \(v \in \var{f}\), then the set of answer sets of \(P\) one-to-one corresponds with the set of minimal trap spaces of \(f\).
When a Boolean formula \(\Phi\) (\(f_v\) or \(\neg f_v\)) is unsafe, the Disjunctive Normal Form (DNF) of $\Phi$ is used instead, which is always safe by definition.

There is no notion of ``safeness'' in the \fasp encoding. 
Rather an additional rule \(\bot \leftarrow \pos{v}, \ngt{v}\) is added to $P$, for every \(v \in \var{f}\) such that the variable \(v\) cannot take the value of \(\star\).
The answer sets of \(P\) one-to-one correspond to the fixed points of \(f\).

Following on, given a BN $f$, we use the notations $\tsconjp{f}$ and $\faspp{f}$ to denote the encoded ASP programs of \(f\), according to the \tsconj and \fasp~encodings, respectively.

\subsection{Methods for Problems \acfirstmts{} and \acfirstfix{}}\label{subsec:methods-first}

We make use of \tsconj and \fasp encodings for~\acfirstmts{} and \acfirstfix{}, respectively. 
The choice of encodings is due to following two reasons:
first, these encodings rely on less expensive representations --- specifically, negation normal forms (ref.\ \Cref{sec:related-work}).
Second, they establish a one-to-one correspondence between the minimal trap spaces (or fixed points) of the original BN and the answer sets of the encoded ASP program, whereas other encodings yield a one-to-one correspondence with the subset-minimal (or subset-maximal) answer sets~\cite{ADFPR2023}, which prevents the direct use of existing ASP counters.

\subsection{Methods for Problems \acsecondmts{} and \acsecondfix{}}\label{subsec:methods-second}

We add the encoding of phenotype to the encodings of \tsconj and \fasp to solve counting problems \acsecondmts{} and \acsecondfix{}, respectively.
Given a BN \(\bn\) and a phenotype \(\phen\),
we compute an ASP program $\toaspalgname{\phen}$ to capture the phenotype \(\phen\) by invoking~\Cref{alg:to_asp}.
The algorithm exploits atoms introduced in the \tsconj encoding to interpret the phenotype $\phen$.
The main idea of~\Cref{alg:to_asp} is exploiting {\em faceted answer set navigation} to constrain the search space of answer sets~\cite{FGR2022}. 
We prove that the minimal trap spaces of \(\bn\) satisfying the phenotype \(\phen\) one-to-one correspond to the answer sets of \(\tsconjp{f} \cup \toaspalgname{\phen}\) (\cref{theo:encoding-correctness-mts}).
\begin{algorithm}[t]
	\caption{$\toaspalgname{\phen}$}
	\label{alg:to_asp}
	\hspace*{\algorithmicindent} \textbf{Input:} Phenotype $\phen$ \\
	\hspace*{\algorithmicindent} \textbf{Output:} ASP program $\phenasp$
	\algnotext{EndFor}
	\algnotext{EndIf}
	\begin{algorithmic}[1]
		\State $\phenasp \gets \emptyset$
		\ForEach {$(v \leftrightarrow e) \in \phen$} 
		\If {$e = 1$}\label{line:condition_vee}
		\State $\phenasp.\mathsf{add}(\bot \leftarrow \dng{\pos{v}}, \qquad \bot \leftarrow \ngt{v})$ \label{line:true_value}
		\ElsIf {$e = 0$} 
		\State $\phenasp.\mathsf{add}(\bot \leftarrow \pos{v}, \qquad \bot \leftarrow \dng{\ngt{v}})$ \label{line:false_value}
		\ElsIf {$e = \star$} 
		\State $\phenasp.\mathsf{add}(\bot \leftarrow \dng{\pos{v}}, \qquad \bot \leftarrow \dng{\ngt{v}})$ \label{line:star_value}
		\EndIf
		\EndFor
		\State \Return {$\phenasp$}
	\end{algorithmic}
\end{algorithm}

\begin{restatable}{theorem}{rcmtscorrectness}\label{theo:encoding-correctness-mts}
	Given a BN \(\bn\) and a phenotype \(\phen\),
	the minimal trap spaces of \(\bn\) satisfying \(\phen\) one-to-one correspond to the answer sets of \(\tsconjp{f} \cup \toaspalgname{\phen}\).
\end{restatable}

For \acsecondfix{}, we can apply the \fasp~encoding and~\Cref{alg:to_asp} similarly.
Note that in the counting problem \acsecondfix{}, the term \(e\) is either $0$ or $1$, for each \((v \leftrightarrow e) \in \phen\), since fixed points require all variables to be fixed.
We formally prove the correctness of our proposed method for \acsecondfix{} in~\Cref{theo:encoding-correctness-fix}.

\begin{restatable}{theorem}{rcfixcorrectness}\label{theo:encoding-correctness-fix}
	Given a BN \(\bn\) and a phenotype \(\phen\),
	the fixed points of \(\bn\) satisfying \(\phen\) one-to-one correspond to the answer sets of \(\faspp{f} \cup \toaspalgname{\phen}\).
\end{restatable}

\begin{example}
	Consider the BN \(f\) of~\Cref{exam:straight-BN}.
	The program $\tsconjp{f}$ is as follows:
	\begin{align*}
		&\pos{a} \vee \ngt{a} \leftarrow \top \quad \quad \pos{a} \leftarrow \pos{a}, \ngt{b} \quad \quad \ngt{a} \leftarrow \aux{1} \quad \quad \aux{1} \leftarrow \ngt{a} \quad \quad \aux{1} \leftarrow \pos{b} \\
		&\pos{b} \vee \ngt{b} \leftarrow \top \quad \quad \pos{b} \leftarrow \pos{a} \quad \quad \ngt{b} \leftarrow \ngt{a}
	\end{align*}
	The program $\tsconjp{f}$ has a unique answer set \(\{\ngt{a}, \ngt{b}, \aux{k}\}\) corresponding to the unique minimal trap space $00$ of \(f\).
	Consider the phenotype \(\phen = (b \leftrightarrow \star)\),
	the BN \(f\) has no minimal trap space satisfying \(\phen\) (see~\Cref{exam:four-first-problems}), thus \(\acsecondmts{}\) returns $0$.
	Following the procedure outlined above, the ASP program \(\toaspalgname{\phen}\) is as follows:
	\begin{align*}
		&\bot \leftarrow \dng{\pos{b}} \qquad \bot \leftarrow \dng{\ngt{b}}
	\end{align*}
	Indeed, the program \(\tsconjp{f} \cup \toaspalgname{\phen}\) has no answer set.
\end{example}

\subsection{Methods for Problems \acthirdmts{} and \acthirdfix{}}\label{subsec:methods-third}

Given a BN \(\bn\), a phenotype \(\phen\), and a set of perturbable variables \(\pert\),
one possible approach to solving these problems is to introduce new atoms to represent possible perturbations over perturbable variables, then correspondingly to intervene in the encoded ASP program obtained by applying the encodings proposed for \acsecondmts{} and \acsecondfix{}, and finally to apply projected counting restricted to these new atoms.
Instead, we propose a more convenient approach that reduces \acthirdmts{} (resp.\ \acthirdfix{}) to \acsecondmts{} (resp.\ \acsecondfix{}) along with projected answer set counting.

\begin{definition}\label{def:BN-perturbation-trans}
	Consider a BN \(f\) and a set of perturbable variables \(\pert \subseteq \var{f}\),
	we construct a new BN \(g\) such that for every \(v \in \var{f}\), if \(v \in \var{f} \setminus \pert\), then the variable \(v \in \var{g}\) and \(g_v = f_v\), and if \(v \in \pert\), then three variables \(v, v^k, v^o \in \var{g}\) and 
	\begin{align*}
		g_v &= \neg v^k \land (v^o \lor f_v), \\
		g_{v^k} &= v^k, \\
		g_{v^o} &= v^o \land \neg v^k.
	\end{align*}
\end{definition}

Instead of modifying the ASP encoding to model perturbations, we construct a perturbed version of the original BN, thereby preserving the semantics of the original encodings.
For each perturbable variable $v \in \pert$, we introduce two new variables in $g$: $v^k$, encoding whether $v$ is knocked out, and $v^o$, encoding whether it is over-expressed. By construction:
\begin{itemize}
	\item If $v^k = 1$, then $v^o = 0$ and $v = 0$, modeling a knockout of $v$.
	\item If $v^k = 0$ and $v^o = 1$, then $v = 1$, modeling over-expression.
	\item If $v^k = 0$ and $v^o = 0$, then $v$ follows its original update function $f_v$ ($v$ is unperturbed).
	\item The case $v^k = 1$ and $v^o = 1$ is infeasible due to the constraint $g_{v^o} = v^o \land \neg v^k$.
\end{itemize}
Since $\var{\bn} \subseteq \var{g}$, any phenotype $\phen$ defined over $\bn$ remains valid in $g$.
The minimal trap spaces (resp. fixed points) of $g$ correspond to those of $f$ under all possible perturbations over $\pert$.
Thus, the minimal trap spaces (resp. fixed points) of $g$ that satisfy the phenotype $\phen$ represent the perturbed solutions of $f$ that also satisfy $\phen$.
Crucially, for each perturbation, multiple satisfying minimal trap spaces (or fixed points) are counted once, enabling us to count the number of satisfying perturbations.
Hence, the counting problem \acthirdmts{} (resp. \acthirdfix{}) reduces to the projected version of \acsecondmts{} (resp. \acsecondfix{}) over the newly defined BN $g$.

We now discuss how we solve \acthirdmts{} and \acthirdfix{} by applying projected answer set counting.
We focus on the case for minimal trap spaces, and the case of fixed points is trivially similar.
Following~\Cref{theo:encoding-correctness-mts}, the set of answer sets of \(\tsconjp{g} \cup \toaspalgname{\phen}\) represents the set of minimal trap spaces of \(g\) satisfying the phenotype \(\phen\).
Let \(\aproj = \bigcup_{v \in \proj}\{\pos{v}, \ngt{v}\}\) denote the set of perturbation-related variables, where $\proj = \bigcup_{v \in \pert}\{v^k, v^o\}$ be the set of ASP atoms used in the encoding.
It follows that the number of answer sets of \(\tsconjp{g} \cup \toaspalgname{\phen}\), projected onto the set \(\aproj\) is equal to  the number of perturbation settings (i.e., assignments to variables in $\pert$) under which $g$ admits a minimal trap space satisfying $\phen$.

\begin{restatable}{theorem}{correctnessperturbationreductionMTS}\label{theo:correctness-perturbation-reduction-MTS}
	Given a BN \(f\), a set of perturbable variables \(\pert \subseteq \var{f}\), a phenotype \(\phen\),
	and \(\aproj = \bigcup_{v \in \proj}\{\pos{v}, \ngt{v}\}\), where $\proj = \bigcup_{v \in \pert}\{v^k, v^o\}$,
	then \acthirdmts{} can be computed as \(\projaspcount{\tsconjp{g} \cup \toaspalgname{\phen}, \aproj}\), where $g$ is the new BN according to~\Cref{def:BN-perturbation-trans}.
\end{restatable}
\begin{restatable}{theorem}{correctnessperturbationreductionFIX}\label{theo:correctness-perturbation-reduction-FIX}
	Given a BN \(f\), a set of perturbable variables \(\pert \subseteq \var{f}\), a phenotype \(\phen\),
	and \(\aproj = \bigcup_{v \in \proj}\{\pos{v}, \ngt{v}\}\), where $\proj = \bigcup_{v \in \pert}\{v^k, v^o\}$,
	then \acthirdfix{} can be computed as \(\projaspcount{\faspp{g} \cup \toaspalgname{\phen}, \aproj}\), where $g$ is the new BN following~\Cref{def:BN-perturbation-trans}.
\end{restatable}

\begin{example}\label{exam:third-MTS-projected}
	Consider again~\Cref{exam:third-MTS}.
	Following~\Cref{def:BN-perturbation-trans}, we obtain the BN \(g\): \(\var{g} = \{a, b, b^k, b^o\}\), \(g_a = a \land \neg b\), \(g_{b^k} = b^k\), \(g_{b^o} = b^o \land \neg b^k\), and \(g_{b} = \neg b^k \land (b^o \lor a)\).
	(see Appendix~\ref{sec:detailed-ASP-programs} for the details of \(\tsconjp{g}\) and \(\toaspalgname{\phen}\)).
	Let \(\aproj = \{\pos{b^k}, \ngt{b^k}, \pos{b^o}, \ngt{b^o}\}\).
	Then \acthirdmts{} can be computed as $\projaspcount{\tsconjp{g} \cup \toaspalgname{\phen}, \aproj}$.
	Indeed, $\projaspcount{\allowbreak \tsconjp{g} \cup \toaspalgname{\phen}, \aproj}$ returns $2$, which is consistent with the result shown in~\Cref{exam:third-MTS}.
\end{example}

%% file: sections/experiment.tex
\section{Experimental Evaluation}\label{sec:experimental-results}

This section presents the experimental evaluation of the presented methods.
We use existing minimal trap space and fixed point computation tools as baselines---namely, AEON~\cite{BKPS2020}, k++ADF (ADF)~\cite{LMNWW22}, and clingo~\cite{GKKOSS2011}.
The ADF tool is applicable here due to the equivalence between ADFs and BNs~\cite{ASDO2024,HMJ2024}, which supports only \acfirstmts{} and \acfirstfix{}.
k++ADF and Clingo count minimal trap spaces and fixed points via enumeration on the ADF and the encoded ASP respectively, and AEON via BDD-based encoding.
For the fixed point problem variants, we further considered the propositional model counter~\cite{SRSM19} (GANAK) and the approximate model counter~\cite{YM2023} (ApproxMC), using the translation to CNF (see~\Cref{sec:related-work}).
However, note that these techniques cannot be directly used for minimal trap space counting. 
For approximate answer set counting, we employed the hashing-based approximate answer set counter ApproxASP~\cite{KESHFM22} with parameters $\varepsilon = 0.8$ and $\delta = 0.2$.
Following prior work on counting~\cite{KM2024}, we provided ApproxASP with an independent support of a disjunctive ASP program exploiting Padoa theorem~\cite{Padoa1901}.
We also tested the tree decomposition-based answer set counter DynASP~\cite{FHMW2017}; however, we did not include it in the final analysis because it has been significantly outperformed by the remaining baselines.
Finally, note that we could not include \#SAT-based ASP counters aspmc~\cite{EHK2024} and sharpASP~\cite{KCM2024} as baselines, since these counters are designed for normal logic programs, while \tsconj~and \fasp~encodings produce disjunctive ASP programs.

We compiled our benchmark set from prior studies on minimal trap spaces and fixed points in BNs~\cite{pastva2023repository,TBPS2024,TBS2023}. The set comprises $645$ BN instances---$245$ real-world models and $400$ randomly generated---with up to $5,\!000$ variables. To evaluate counting problems \acsecondmts{}, \acsecondfix{}, \acthirdmts{}, and \acthirdfix{}, we pseudo-randomly fixed three variables to represent the target phenotype and selected up to $50$ perturbable variables (yielding as many as $3^{50}$ possible perturbations).
Phenotypes are typically not published in machine-readable format, thus we maintain biological interpretability by deriving the phenotype from a known trap space, linking it to an existing biological feature of the network.
As for the chosen size, only a few variables are sufficient to identify phenotypes: e.g., in~\cite{Fischer2021}, only 10-200 entities are needed out of ~10,000. 
Since our tests are often significantly smaller, we scaled down the phenotypes accordingly.
Appendix~\ref{sec:detailed-benchmark-inputs} provides further details about the benchmark.
The code and dataset of experiment evaluation is available at: \url{https://zenodo.org/records/15141045}

\textbf{Environmental Settings.} All experiments were conducted on a high-performance computing cluster, with each node consisting of Intel Xeon Gold $6248$ CPUs. 
Each benchmark instance was allocated one core, with runtime and memory limits set to $5000$ seconds and $8$ GB respectively, for all the tools considered.

\subsection{Experimental Results}

\textbf{\acfirstmts{} and \acfirstfix{}} The results for \acfirstmts{} and \acfirstfix{} are shown in~\Cref{table:c_mts_1_result} and~\Cref{table:c_fix_1_result}, respectively.
Each table reports the number of instances solved (i.e., instances for which a count was successfully returned) by each tool and their corresponding PAR$2$ scores~\cite{SAT2017} (PAR$2$ score is a runtime metric that also penalizes benchmark timeouts). 
Here, approximate counting (ApproxASP and ApproxMC) clearly outperform all existing solutions. 
Even compared to exact model counting (GANAK), this approach achieves significantly better performance. 
Note that for \acfirstfix{}, ApproxASP and ApproxMC are roughly comparable, but (as discussed in detail later), ApproxASP is faster on simpler instances.

It is worth noting that unsafe formulas are quite rare in 245 real-world models, which is consistent with the observation in~\cite{TBPS2024}.
All 400 randomly generated models have no unsafe formulas because of the nature of the generation~\cite{TBPS2024}.

\begin{table}[h]
    \centering
    \begin{tabular}{m{4em} m{4em} m{4em} m{4em} m{5em}} 
    \toprule
    & AEON & ADF & clingo & ApproxASP\\
    \midrule
    \#Solved & 179 & 200 & 211 & \textbf{364}\\
    \midrule
    PAR$2$ & 7255 & 6923 & 6742 & \textbf{4448}\\
    \bottomrule
    \end{tabular}
    \caption{The performance comparison of different counters on \acfirstmts{} counting problem.}
    \label{table:c_mts_1_result}
\end{table}
\begin{table}[h]
    \centering
    \begin{tabular}{m{4em} m{3em} m{3em} m{3em} m{4em} m{5em} m{5em}} 
    \toprule
    & AEON & ADF & clingo & GANAK & ApproxMC & ApproxASP \\
    \midrule
    \#Solved & 247 & 217 & 227 & 317 & \textbf{420} & 413\\
    \midrule
    PAR$2$ & 6172 & 6656 & 6493 & 5269 & 3801 & \textbf{3760}\\
    \bottomrule
    \end{tabular}
    \caption{The performance comparison of different counters on \acfirstfix{} counting problem.}
    \label{table:c_fix_1_result}
\end{table}

\begin{figure*}
    \centering
    \begin{subfigure}[t]{0.32\textwidth}
        \centering
        \includegraphics[width=0.99\linewidth]{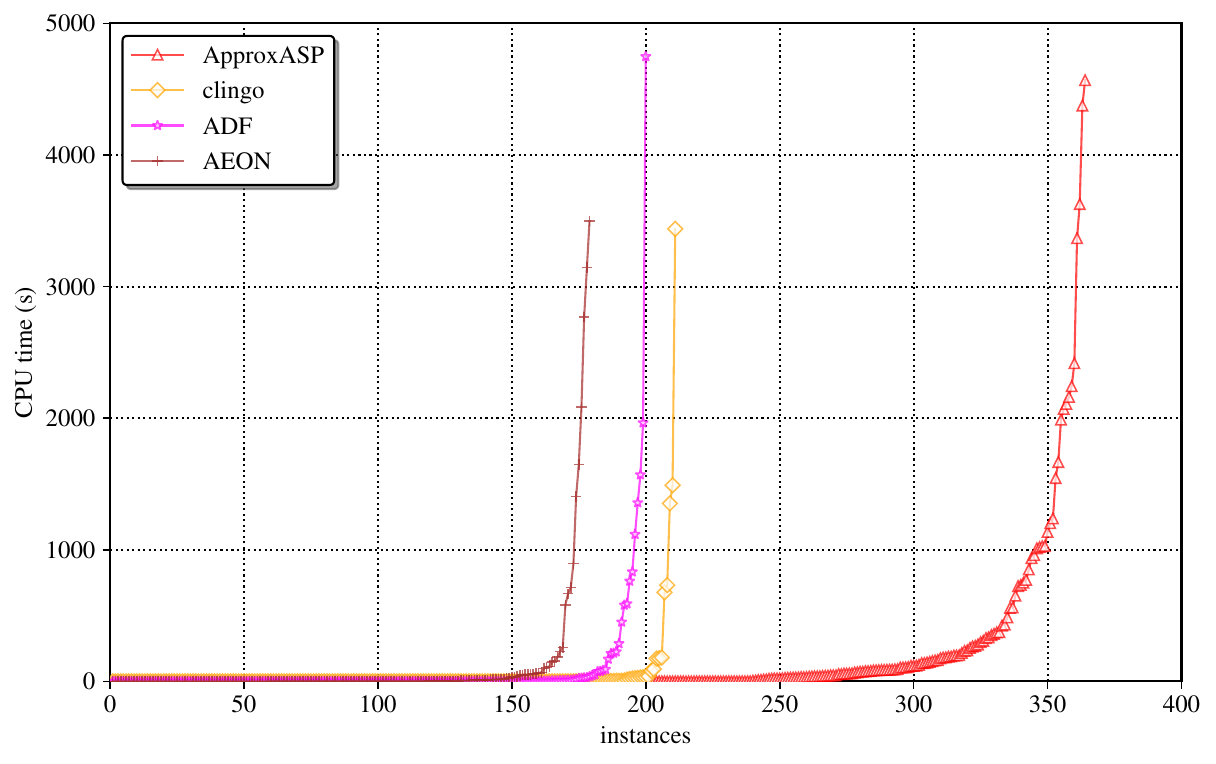}
        \caption{\acfirstmts{}}
        \label{fig:count_mts_1}
    \end{subfigure}
    \begin{subfigure}[t]{0.32\textwidth}
        \centering
        \includegraphics[width=0.99\linewidth]{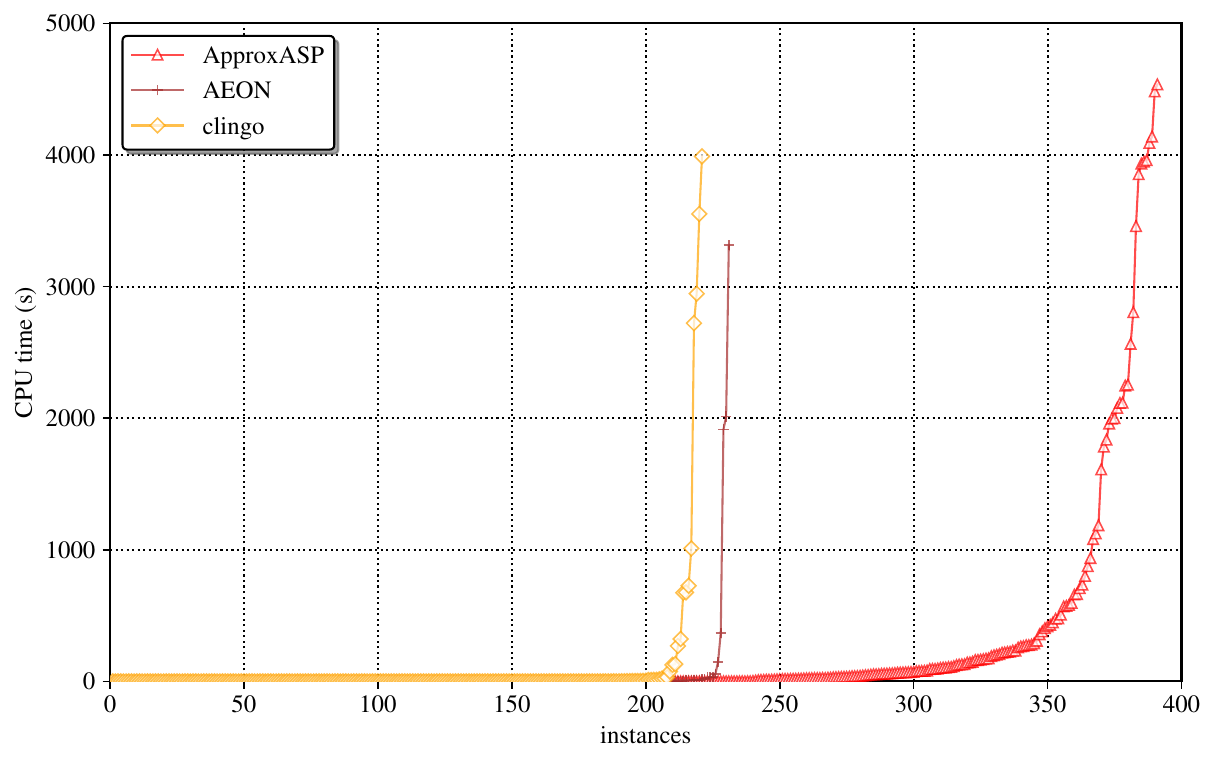}
        \caption{\acsecondmts{}}
        \label{fig:count_mts_2}
    \end{subfigure}
    \begin{subfigure}[t]{0.32\textwidth}
        \centering
        \includegraphics[width=0.99\linewidth]{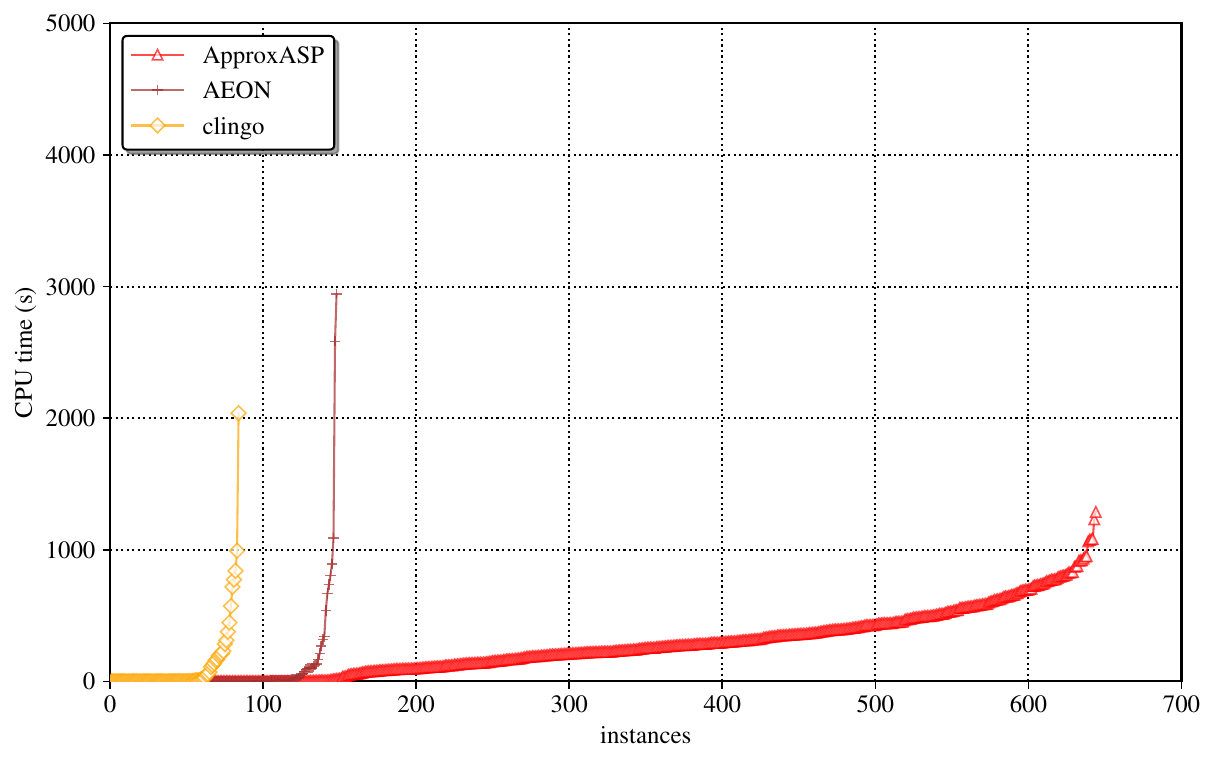}
        \caption{\acthirdmts{}}
        \label{fig:cactus_plots}
    \end{subfigure}
    \begin{subfigure}[t]{0.32\textwidth}
        \centering
        \includegraphics[width=0.99\linewidth]{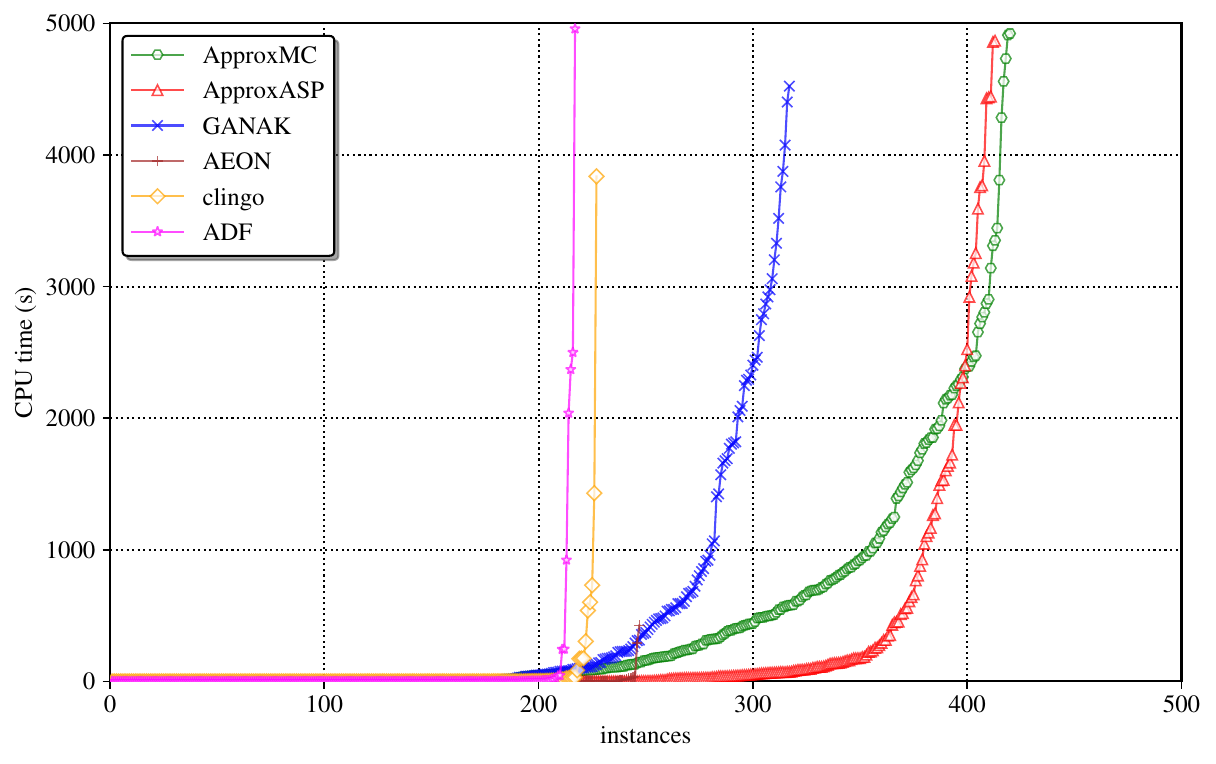}
        \caption{\acfirstfix{}}
        \label{fig:count_fix_1}
    \end{subfigure}
    \begin{subfigure}[t]{0.32\textwidth}
        \centering
        \includegraphics[width=0.99\linewidth]{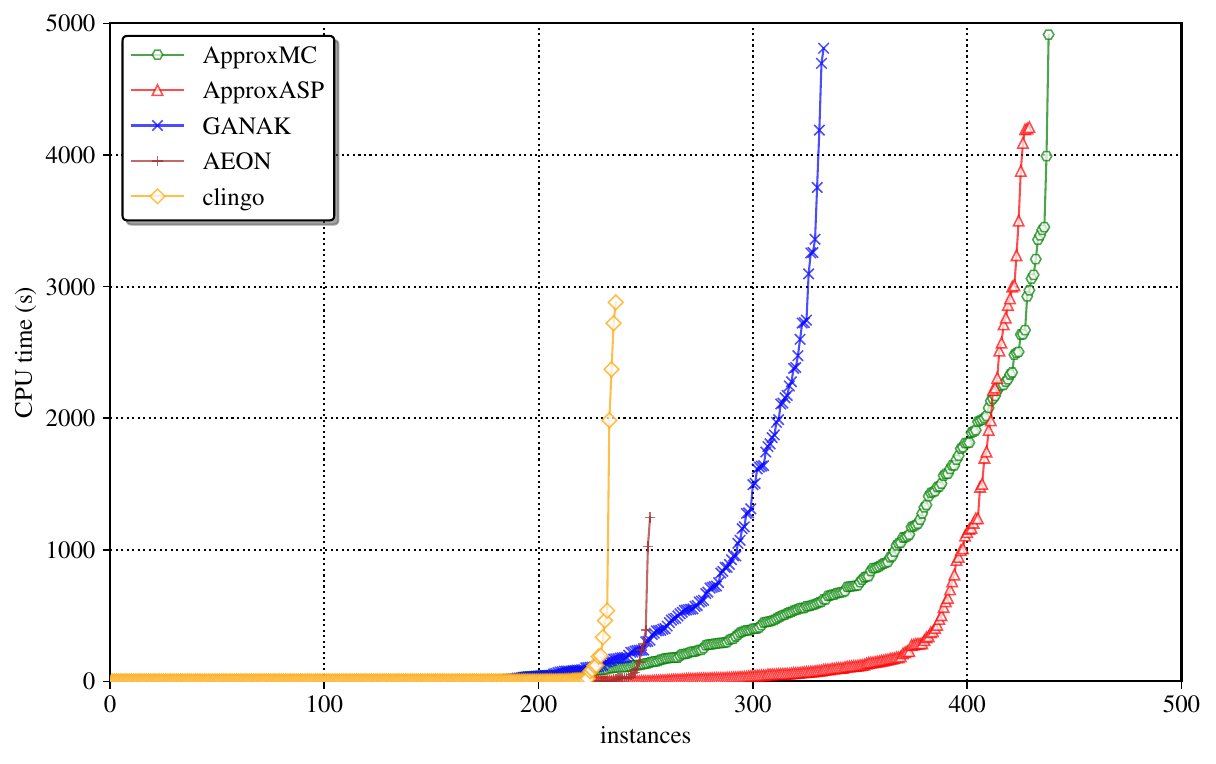}
        \caption{\acsecondfix{}}
        \label{fig:count_fix_2}
    \end{subfigure}
    \begin{subfigure}[t]{0.32\textwidth}
        \centering
        \includegraphics[width=0.99\linewidth]{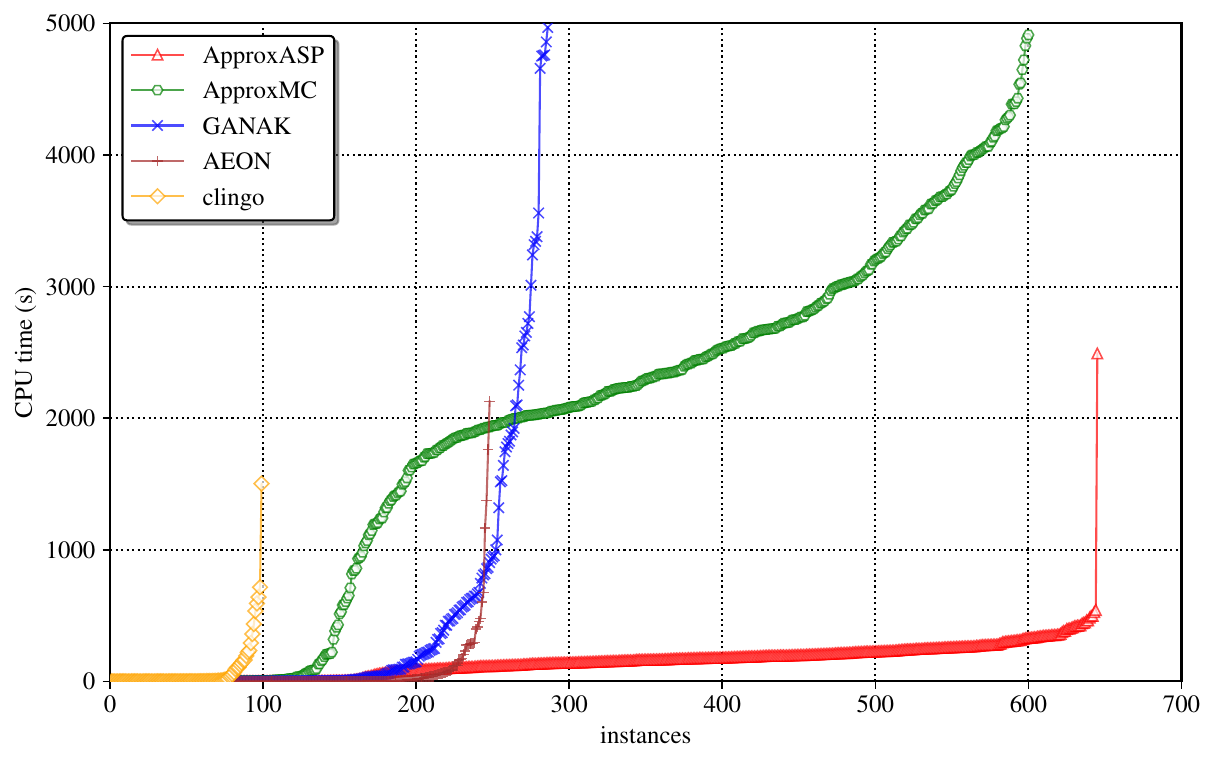}
        \caption{\acthirdfix{}}
        \label{fig:count_fix_3}
    \end{subfigure}
    
    \caption{Performance comparison of different counters across all counting problems. The $x$ axis shows the number of benchmarks completed before the corresponding CPU runtime on the $y$ axis.}
    \label{fig:counting_task_1}
\end{figure*}

\noindent\textbf{\acsecondmts{} and \acsecondfix{}} The results for \acsecondmts{} and \acsecondfix{} are shown in~\Cref{table:c_mts_2_result} and~\Cref{table:c_fix_2_result}, respectively, using the same metric (solved instances and PAR$2$ score) as~\Cref{table:c_mts_1_result} and~\Cref{table:c_fix_1_result}.
Here, the relative performance of individual tools mirrors that observed for \acfirstmts{} and \acfirstfix{}.
However, all tools solved more instances overall, largely due to the inclusion of the phenotype property, which typically reduces the number of solutions.
Since phenotype properties are generally much simpler than the update functions describing network dynamics, we expect them to simplify the counting problem---often substantially.

\begin{table}[h]
    \centering
    \begin{tabular}{m{4em} m{4em} m{4em} m{5em}} 
    \toprule
    & AEON & clingo & ApproxASP\\
    \midrule
    \#Solved & 231 & 308 & \textbf{464}\\
    \midrule
    PAR$2$ & 6428 & 5237 & \textbf{2919}\\
    \bottomrule
    \end{tabular}
    \caption{The performance comparison of different counters on \acsecondmts{} counting problem.}
    \label{table:c_mts_2_result}
\end{table}

\begin{table}[h]
    \centering
    \begin{tabular}{m{4em} m{3em} m{3em} m{4em} m{5em} m{5em}} 
    \toprule
    & AEON & clingo & GANAK & ApproxMC & ApproxASP\\
    \midrule
    \#Solved & 252 & 236 & 333 & \textbf{438} & 429\\
    \midrule
    PAR$2$ & 6099 & 6360 & 5030 & 3527 & \textbf{3499}\\
    \bottomrule
    \end{tabular}
    \caption{The performance comparison of different counters on \acsecondfix{} counting problem.}
    \label{table:c_fix_2_result}
\end{table}

\medskip

\noindent\textbf{\acthirdmts{} and \acthirdfix{}} The results for problems \acthirdmts{} and \acthirdfix{} are presented in~\Cref{table:c_mts_3_result,table:c_fix_3_result}, respectively. 
This benchmark confirms the leading performance of ApproxASP in minimal trap space and fixed point counting as it was able to solve $644/645$ and $645/645$ problem instances for \acthirdmts{} and \acthirdfix{}, respectively.
Here, ApproxASP significantly outperforms even ApproxMC, and outperforms all other tools by a factor of $2\times$ or more.
In contrast to \acsecondmts{} and \acsecondfix{}, where all tools benefited from a reduced number of solutions, the presence of perturbations in this setting generally increases both the number of solutions and the underlying complexity of BNs.
The key performance differentiation lies in the tools' ability to handle \emph{projected counting}.
For this problem, the {\em independent support} for XOR constraints is derived from the BN perturbable variables. 
Since this is only a subset of the network variables, the independent support size is relatively small, reducing the size of the XORs. 
This results in the superior performance of ApproxASP and ApproxMC.
Note that in these countings, the number of perturbable variables is at most $50$ and the count is upper-bounded by $3^{50}$.

\begin{table}[h]
    \centering
    \begin{tabular}{m{4em} m{4em} m{4em} m{5em}} 
    \toprule
    & AEON & clingo & ApproxASP\\
    \midrule
    \#Solved & 148 & 84 & \textbf{644}\\
    \midrule
    PAR$2$ & 7725 & 8711 & \textbf{283}\\
    \bottomrule
    \end{tabular}
    \caption{The performance comparison of different counters on \acthirdmts{} counting problem.}
    \label{table:c_mts_3_result}
\end{table}
\begin{table}[h]
    \centering
    \begin{tabular}{m{4em} m{3em} m{3em} m{4em}  m{5em} m{5em}} 
    \toprule
    & AEON & clingo & GANAK & ApproxMC & ApproxASP\\
    \midrule
    \#Solved & 248 & 99 & 286 & 600 & \textbf{645}\\
    \midrule
    PAR$2$ & 6176 & 8476 & 5757 & 2481 & \textbf{150}\\
    \bottomrule
    \end{tabular}
    \caption{The performance comparison of different counters on \acthirdfix{} counting problem.}
    \label{table:c_fix_3_result}
\end{table}

\subsection{Detailed Runtime Analysis}

The runtime performance of different tools is depicted in~\Cref{fig:counting_task_1}.
In these cactus plots, a point $(x,y)$ indicates that a tool successfully completes $x$ benchmark instances, with each instance taking at most $y$ seconds. 
The plots highlight the superiority of the hashing-based counting techniques, ApproxASP and ApproxMC. 
Notably, even in cases where ApproxASP solves fewer instances than ApproxMC (e.g. \acfirstfix{}), it is typically faster on simpler problem instances, which is also reflected in its PAR$2$ score.

By examining the number of solutions successfully computed for different tasks, we observe that only ApproxASP, ApproxMC, and GANAK can reliably count instances having a large number of solutions (e.g. $\geq 10^{30}$). 
Here, BDD-based counters like AEON perform somewhat better on fixed point problems compared to tools using plain enumeration (ADF, clingo), but cannot compete in the (arguably more complex) trap space problems.

The performance of GANAK and ApproxMC is also severely affected by the time required to compute their input CNF formulas (ref. \Cref{sec:related-work}). 
Here, deriving the CNF problem representation is often considerably more time-consuming than computing the comparable ASP encoding.
Our results reveal that, on average, it took about $545$ seconds to compute the CNF formula for each BN.
Moreover, for $39$ BNs, the corresponding CNF could not be computed within the $5000$-second timeout.
In contrast, the ASP encodings for all BNs—including the more challenging ones with perturbable variables—were generated within seconds.
This demonstrates the superior flexibility of the ASP-based approach.

Note that the approximate counters ApproxMC and ApproxASP provide an $(\varepsilon, \delta)$-guarantee. 
Thus for each solved instance, we computed the {\em observed tolerance}, which is defined as $\mathsf{max}(\sfrac{\cnt}{\Card{\as{P}}}, \sfrac{\Card{\as{P}}}{\cnt}) - 1$, where $\cnt$ is the count returned by ApproxASP or ApproxMC, and $\as{P}$ denotes the answer set count of program $P$.
On average, ApproxMC and ApproxASP exhibit observed tolerances of $0.032$ and $0.007$, respectively.
The maximum observed tolerances were $0.39$ for ApproxASP and $0.07$ for ApproxMC, both of which are well below the theoretical bound of $\varepsilon = 0.8$.

We evaluated ApproxMC and ApproxASP with a tighter guarantee, setting $\varepsilon = 0.01$ and $\delta = 0.05$.
In these setting, ApproxASP solved $244, 243, 114, 258, 259, 130$ for \acfirstmts{}, \acsecondmts{}, \acthirdmts{}, \acfirstfix{}, \acsecondfix{}, and \acthirdfix{}, respectively, 
while ApproxMC solved $263, 259, 120$ for \acfirstfix{}, \acsecondfix{}, and \acthirdfix{}, respectively. 
Consequently, the higher precision significantly reduces the number of solved instances. 

We compared the size of instances solved by different counting techniques across various counting problems. 
We observed that, for minimal trap spaces, ApproxASP solved instances with up to $4000$ variables, while \tsconj~and AEON solved instances with up to $321$ variables.
For fixed points, although there were large instances (up to $4000$ variables), these did not contain fixed points, making a direct comparison unfeasible.

%% file: sections/conclusion.tex
\section{Conclusion}\label{sec:conclusion}

This paper addresses the problem of counting minimal trap spaces and fixed points in Boolean networks.
These are critical concepts in understanding of long-term BN behavior and are relevant across a diverse set of application domains, including systems biology, abstract argumentation, and logic programming. Trap space counting is especially important in systems biology: due to the inherent robustness of biological phenomena, biologically motivated BNs admit a high degree of redundancy, resulting in a vast repertoire of closely related trap spaces (or fixed points) that cannot be explored solely through enumeration.

Here, we propose novel methods for determining trap space and fixed point counts using approximate answer set counting, thus entirely avoiding costly enumeration. We apply this methodology to three biologically motivated problems: (a) general counting; (b) counting occurrences of a known biological phenotype, and (c) projected counting of perturbations that ensure the emergence of a known biological phenotype. The last problem is particularly timely, as it allows us to determine \emph{perturbation robustness}~\cite{BBPSS2023,kitano2007towards}, a vital measure that determines how stable a phenotype appears under external stimuli. Through extensive experiments on a diverse set of benchmarks, we show that approximate counting substantially improves the feasibility of counting in this domain, outperforming traditional enumeration-based and exact approaches whenever applicable.

Our work opens several promising directions for future research:
First is to integrate reduction techniques---previously shown to be effective in BN analysis~\cite{Rozum2021,ST23,TP2024}---to further improve counting accuracy and scalability. 
Second direction explores hybrid strategies that combine exact and approximate counting, aiming to strike a balance between efficiency and precision. 
Finally, a deeper investigation into the computational complexity of the counting problems would help refine our understanding of their theoretical underpinnings.

%% file: sections/appendix-a.tex
\section{Details of Proofs}\label{sec:detailed-proofs}

\rcmtscorrectness*
\begin{proof}
	To prove the correctness of~\Cref{theo:encoding-correctness-mts}, we reuse the correctness proof of Theorem $2$ of~\cite{TBPS2024}, which establishes that the answer sets of \(\tsconjp{f}\) one-to-one correspond to the minimal trap spaces of $\bn$.
	Let us recall the translation between an answer set \(M\) of \(\tsconjp{f}\) and its respective minimal trap space \(m\) of \(f\): for each variable $v \in \var{\bn}$, \(m(v) = 1\) if and only if \(\pos{v} \in M \land \ngt{v} \not \in M\), \(m(v) = 0\) if and only if \(\pos{v} \not \in M \land \ngt{v} \in M\), and \(m(v) = \star\) if and only if \(\pos{v} \in M \land \ngt{v} \in M\).

	We employ the theories of faceted answer set navigation~\cite{FGR2022} to prove the correctness of~\Cref{theo:encoding-correctness-mts}.
	According to faceted navigation, for a program $P$ and an atom $a \in \at{P}$,  adding the integrity constraint $\{\bot \leftarrow a\}$ to program $P$ restricts the search space of $P$, where no answer set contains the atom $a$.
	Conversely, adding the integrity constraint $\{\bot \leftarrow \dng{a}\}$ to program $P$ ensures that every answer set in the modified program contains the atom $a$.

	We prove the correctness of our encoding for each trait $(v \leftrightarrow e) \in \phen$.
	We show that~\Cref{alg:to_asp} selects {\em facets} in such a way that the answer sets of $\tsconjp{f} \cup \toaspalgname{\phen}$ one-to-one correspond to the minimal trap spaces satisfying $\phen$ of \(f\).
	When $(v \leftrightarrow 1) \in \phen$ (line~\ref{line:true_value} in~\Cref{alg:to_asp}), two constraints are added to $\toaspalgname{\phen}$ to ensure that every answer set contains the atom $\pos{v}$ and none contains the atom $\ngt{v}$.
	These constraints effectively assign the value $1$ to variable $v$.
	When $(v \leftrightarrow 0) \in \phen$ (line~\ref{line:false_value} in~\Cref{alg:to_asp}), 
	two added constraints ensure that every answer set contains the atom $\ngt{v}$ and none contains $\pos{v}$, thereby assigning the value $0$ to variable $v$.
	When $(v \leftrightarrow \star) \in \phen$ (line~\ref{line:star_value} in~\Cref{alg:to_asp}), 
	two added constraints ensure that every answer set contains both $\pos{v}$ and $\ngt{v}$. 
	These constraints effectively assign the value $\star$ to variable $v$.

	\noindent Combining all the cases of~\Cref{alg:to_asp}, we can claim the correctness of our encoding.
\end{proof}

\rcfixcorrectness*
\begin{proof}
	The proof technique of~\Cref{theo:encoding-correctness-mts} can be similarly extended for fixed point counting with the program $\faspp{f}$.
\end{proof}

To prove Theorems~\ref{theo:correctness-perturbation-reduction-MTS} and~\ref{theo:correctness-perturbation-reduction-FIX}, we prepare the following preliminaries.

\begin{definition}\label{def:min-t-po}
	The total order \(\leq_t\) on \(\threed\) is defined by \(0 <_t \star <_t 1\).
\end{definition}

\begin{definition}\label{def:min-s-po}
	The partial order \(\leq_s\) on \(\threed\) is defined by \(0 <_s \star\), \(1 <_s \star\), and it contains no other relation.
\end{definition}

\begin{definition}\label{def:three-valued-evaluation}
	Consider a BN \(\bn\), a sub-space \(m\) of \(\bn\), and a Boolean expression \(e\) over \(\var{\bn}\).
	The value of \(e\) under sub-space \(m\) w.r.t. the Kleene three-valued logic, denoted as \(m(e)\), is recursively defined as follows:
	\begin{align*}
		m(e) = \begin{cases}
			e &\text{if } e \in \threed\\
			m(a) &\text{if } e = a, a \in \var{\bn}\\
			\neg m(e_1) &\text{if } e = \neg e_1\\
			\text{min}_{\leq_t}(m(e_1), m(e_2)) &\text{if } e = e_1 \land e_2\\
			\text{max}_{\leq_t}(m(e_1), m(e_2)) &\text{if } e = e_1 \lor e_2
		\end{cases}
	\end{align*} where \(\neg 1 = 0, \neg 0 = 1, \neg \star = \star\), and \(\text{min}_{\leq_t}\) (resp.\ \(\text{max}_{\leq_t}\)) is the function to get the minimum (resp.\ maximum) value of two values w.r.t.\ the order \(\leq_t\).
\end{definition}

\begin{theorem}[Theorem 1 of~\cite{HAH2015}]\label{theo:trap-space-char}
	Consider a BN \(f\) and a sub-space \(m\) of \(f\).
	A sub-space \(m\) is a trap space of \(f\) iff \(m(f_v) \leq_s m(v)\) for every \(v \in \var{f}\).
\end{theorem}

\begin{corollary}\label{cor:mts-char}
	Consider a BN \(\bn\) and a sub-space \(m\) of \(\bn\).
	A sub-space \(m\) is a minimal trap space of \(\bn\) iff \(m\) is a \(\leq_s\)-minimal trap space of \(\bn\).
\end{corollary}

\begin{corollary}\label{cor:fix-point-char}
	Consider a BN \(\bn\) and a sub-space \(m\) of \(\bn\).
	A sub-space \(m\) is a fixed point of \(\bn\) iff \(m\) is a trap space of \(\bn\) and \(m(v) \neq \star\) for every \(v \in \var{\bn}\).
\end{corollary}

Now, we show the formal proof of Theorem~\ref{theo:correctness-perturbation-reduction-MTS}.

\begin{lemma}\label{lem:transformed-BN-mts-fixed-value}
	Consider a BN \(f\) and a set of perturbable variables \(\pert \subseteq \var{f}\).
	Let \(g\) be the BN obtained from \(f\), according to~\Cref{def:BN-perturbation-trans}.
	Let \(\proj \) be the set \( \bigcup_{v \in \pert}\{v^k, v^o\}\).
	If \(m\) is a minimal trap space of \(g\), then \(m(v) \neq \star\) for every \(v \in \proj\).
\end{lemma}
\begin{proof}
	Let \(m\) be a minimal trap space of \(g\).
	Assume that there exists \(v^k \in \proj\) such that \(m(v^k) = \star\) or  \(v^o \in \proj\) such that \(m(v^o) = \star\).
	We consider two cases as follows.
	
	\textbf{Case 1}: \(m(v^k) = \star\).
	Let \(m'\) be a sub-space of \(g\) such that \(m'(u) = m(u)\), for every \(u \in \var{g} \setminus \{v^k\}\) and \(m'(v^k) = 0\).
	We have \(m'(g_{v^k}) = m'(v^k)\).
	The variable \(v^k\) only affects \(v^o\) and \(v\).
	Regarding \(v^o\), we have \(m'(g_{v^o}) = m'(v^o \land \neg v^k) = \text{min}_{\leq_t}(m'(v^o), \neg m'(v^k)) = \text{min}_{\leq_t}(m'(v^o), 1) = m'(v^o)\).
	Regarding \(v\), we have \(m'(g_v) = m'(\neg v^k \land (v^o \lor f_v)) = \text{min}_{\leq_t}(\neg m'(v^k), \text{max}_{\leq_t}(m'(v^o), m'(f_v))) = \text{min}_{\leq_t}(1, \text{max}_{\leq_t}(m'(v^o), m'(f_v))) = \text{max}_{\leq_t}(m'(v^o),\allowbreak m'(f_v)) =  \text{max}_{\leq_t}(m(v^o), m(f_v))\).
	Following~\Cref{theo:trap-space-char}, \(m(g_v) \leq_s m(v)\).
	We have \(m(g_v) = m(\neg v^k \land (v^o \lor f_v)) = \text{min}_{\leq_t}(\neg m(v^k), \text{max}_{\leq_t}\allowbreak(m(v^o), m(f_v))) = \text{min}_{\leq_t}(\star, \text{max}_{\leq_t}(m(v^o), m(f_v)))\).
	Since \(m(v) <_s \text{max}_{\leq_t}(m(v^o), m(f_v))\) implies \(m(v) <_s m(g_v)\) which is a contradiction, we derive that \(\text{max}_{\leq_t}(m(v^o), m(f_v)) \leq_s m(v)\).
	This implies \(m'(g_v) \leq_s m(v) = m'(v)\).
	For any \(u \in \var{g} \setminus \{v, v^k, v^o\}\), \(m'(g_u) = m(g_u) \leq_s m(u) = m'(u)\).
	Hence, \(m'\) is trap space of \(g\) and \(m' <_s m\), which is a contradiction.
	
	\textbf{Case 2}: \(m(v^o) = \star\).
	Let \(m'\) be a sub-space of \(g\) such that \(m'(u) = m(u)\) for every \(u \in \var{g} \setminus \{v^o\}\) and \(m'(v^o) = 0\).
	We have \(m'(g_{v^o}) = m'(v^o \land \neg v^k) = \text{min}_{\leq_t}(m'(v^o), \neg m'(v^k)) = \text{min}_{\leq_t}(0, \neg m'(v^k)) = 0 = m'(v^o)\).
	The variable \(v^o\) only affects \(v\).
	We have \(m'(g_v) = m'(\neg v^k \land (v^o \lor f_v)) = \text{min}_{\leq_t}(\neg m'(v^k), \text{max}_{\leq_t}(m'(v^o), m'(f_v))) = \text{min}_{\leq_t}(\neg m'(v^k), \text{max}_{\leq_t}(0, m'(f_v))) = \text{min}_{\leq_t}(\neg m'(v^k), m'(f_v)) = \text{min}_{\leq_t}(\neg m(v^k), m(f_v))\).
	Following~\Cref{theo:trap-space-char}, \(m(g_v) \leq_s m(v)\).
	We have \(m(g_v) = m(\neg v^k \land (v^o \lor f_v)) = \text{min}_{\leq_t}(\neg m(v^k), \text{max}_{\leq_t}(m(v^o), m(f_v))) = \text{min}_{\leq_t}(\neg m(v^k),\allowbreak \text{max}_{\leq_t}(\star, m(f_v)))\).
	Suppose that \(m'(v) = m(v) <_s \text{min}_{\leq_t}(\neg m(v^k), m(f_v))\).
	If \(m(f_v) = 0\), then \(m(v) <_s 0\) which contradicts to the definition of \(\leq_s\).
	Hence \(m(f_v) \neq 0\), leading to \(\text{max}_{\leq_t}(\star, m(f_v)) = m(f_v)\).
	Then \(m(v) <_s \text{min}_{\leq_t}(\neg m(v^k), m(f_v)) = m(g_v)\), which is a contradiction.
	Hence, \(\text{min}_{\leq_t}(\neg m(v^k), m(f_v)) \leq_s m'(v)\).
	This implies that \(m'(g_v) \leq_s m'(v)\).
	For any \(u \in \var{g} \setminus \{v, v^o\}\), \(m'(g_u) = m(g_u) \leq_s m(u) = m'(u)\).
	Hence, \(m'\) is trap space of \(g\) and \(m' <_s m\), which is a contradiction.
	
	Combining Case $1$ and Case $2$, we can conclude that for \(m(v) \neq \star\) for every \(v \in \proj\).
\end{proof}

\begin{lemma}\label{lem:trap-space-projection}
	Consider a BN \(f\) and a set of perturbable variables \(\pert \subseteq \var{f}\).
	Let \(g\) be the BN obtained from \(f\), according to~\Cref{def:BN-perturbation-trans} and \(\proj\) be the set \(\bigcup_{v \in \pert}\{v^k, v^o\}\).
	Let \(m\) be a trap space of \(g\) such that \(m(v) \neq \star\) for every \(v \in \proj\).
	Let \(\perm\) be the perturbation of \(f\) such that \(\perm(v) = 1\) if and only if \(m(v^k) = 0\) and \(m(v^o) = 1\), \(\perm(v) = 0\) if and only if \(m(v^k) = 1\) and \(m(v^o) = 0\), and \(\perm(v) = \star\) if and only if \(m(v^k) = 0\) and \(m(v^o) = 0\).
	Then \(m'\) is a trap space of \(f^{\perm}\) where \(m'(v) = m(v)\) for every \(v \in \var{f}\).
\end{lemma}
\begin{proof}
	Assume that there exists \(v \in \pert\) such that \(m(v^k) = m(v^o) = 1\).
	Since \(m\) is a trap space of \(g\), \(m(g_{v^o}) \leq_s m(v^o)\).
	We have \(m(g_{v^o}) = m(v^o \land \neg v^k) = 0\), whereas \(m(v^o) = 1\), leading to \(0 \leq_s 1\), which contradicts to the definition of \(\leq_s\).
	Hence, the case of \(m(v^k) = m(v^o) = 1\) is impossible for every \(v \in \pert\).
	Since \(m(v) \neq \star\) for every \(v \in \proj\), \(\perm\) is well specified.
	
	Recall that \(\var{g} = \var{f} \cup \proj\) and \(\var{f} = \var{f^{\perm}}\).
	Consider \(v \in \var{f^{\perm}}\).
	If \(v \not \in \pert\), we have \(m'(f^{\perm}_v) = m'(f_v) = m'(g_v) = m(g_v) \leq_s m(v) = m'(v)\).
	If \(v \in \pert\), we have the following cases:
	
	\textbf{Case 1}: \(m(v^k) = 0\) and \(m(v^o) = 0\).
	Then \(\perm(v) = \star\), thus \(m'(f^{\perm}_v) = m'(f_v)\).
	We have \(m(g_v) = m(\neg v^k \land (v^o \lor f_v)) = m(f_v) \leq_s m(v) = m'(v)\).
	Since \(m'(f_v) = m(f_v)\), it follows that \(m'(f^{\perm}_v) \leq_s m'(v)\).
	
	\textbf{Case 2}: \(m(v^k) = 1\) and \(m(v^o) = 0\).
	Then \(\perm(v) = 0\), thus \(m'(f^{\perm}_v) = 0\).
	We have \(m(g_v) = m(\neg v^k \land (v^o \lor f_v)) = 0 \leq_s m(v) = m'(v)\).
	Hence, \(m'(f^{\perm}_v) \leq_s m'(v)\).
	
	\textbf{Case 3}: \(m(v^k) = 0\) and \(m(v^o) = 1\).
	Then \(\perm(v) = 1\), thus \(m'(f^{\perm}_v) = 1\).
	We have \(m(g_v) = m(\neg v^k \land (v^o \lor f_v)) = 1 \leq_s m(v) = m'(v)\).
	Hence, \(m'(f^{\perm}_v) \leq_s m'(v)\).
	
	Now we can conclude that \(m'(f^{\perm}_v) \leq_s m'(v)\) for every \(v \in \var{f^{\perm}}\).
	Hence, \(m'\) is a trap space of \(f^{\perm}\).
\end{proof}

\correctnessperturbationreductionMTS*
\begin{proof}
	First, we consider a perturbation \(\perm \colon \pert \to \threed\) of BN \(f\).
	Let \(m_{\proj} \colon \proj \to \twod\) be a mapping such that for every \(v \in \pert\), \(\perm(v) = 1\) if and only if \(m_{\proj}(v^k) = 0\) and \(m_{\proj}(v^o) = 1\), \(\perm(v) = 0\) if and only if \(m_{\proj}(v^k) = 1\) and \(m_{\proj}(v^o) = 0\), and \(\perm(v) = \star\) if and only if \(m_{\proj}(v^k) = 0\) and \(m_{\proj}(v^o) = 0\).
	Recall that \(\var{g} = \var{f} \cup \proj\) and \(\var{f} = \var{f^{\perm}}\).
	Let \(m\) be a minimal trap space of \(f^{\perm}\) and \(m \models \phen\).
	Let \(m'\) be a sub-space of \(g\) such that \(m'(v) = m(v)\) if \(v \in \var{f}\) and \(m'(v) = m_{\proj}(v)\) if \(v \in \proj\).
	We show that \(m'\) is a minimal trap space of \(g\) and \(m' \models \phen\) (1).
	
	Consider \(v \in \pert\),
	we have \(m'(g_{v^k}) = m'(v^k)\).
	If \(m'(v^k) = 0\), then \(m'(v^o \land \neg v^k) = m'(v^o)\).
	If \(m'(v^k) = 1\), then \(m'(v^o) = 0\) due to the definition of \(m_{\proj}\), leading to \(m'(v^o \land \neg v^k) = 0 = m'(v^o)\).
	Hence, we can derive that \(m'(g_{v^o}) = m'(v^o \land \neg v^k) = m'(v^o)\).
	Regarding \(m'(g_v) = m'(\neg v^k \land (v^o \lor f_v))\), we have the following cases:
	
	\textbf{Case 1}: \(\perm(v) = \star\).
	Then \(m'(v^k) = 0\) and \(m'(v^o) = 0\).
	We have \(m'(g_v) = m'(f_v) = m(f_v) = m(f^{\perm}_v) \leq_s m(v) = m'(v)\).

	\textbf{Case 2}: \(\perm(v) = 1\).
	Then \(m'(v^k) = 0\) and \(m'(v^o) = 1\).
	We have \(m'(g_v) = 1 = m(f^{\perm}_v) \leq_s m(v) = m'(v)\).

	\textbf{Case 3}: \(\perm(v) = 0\).
	Then \(m'(v^k) = 1\) and \(m'(v^o) = 0\).
	We have \(m'(g_v) = 0 = m(f^{\perm}_v) \leq_s m(v) = m'(v)\).
	Consider \(v \in \var{f} \setminus \pert\).
	We have \(m'(g_v) = m'(f_v) = m(f_v) = m(f^{\perm}_v) \leq_s m(v) = m'(v)\).
	
	Now, we can conclude that \(m'\) is a trap space of \(g\).
	Assume that \(m'\) is not minimal.
	Then there is a trap space \(n\) of \(g\) such that \(n <_s m'\).
	Since \(m'(v) \neq \star\) for every \(v \in \proj\), \(n(v) = m'(v)\) for every \(v \in \proj\), leading to \(n(v) \neq \star\) for every \(v \in \proj\).
	Following the~\Cref{lem:trap-space-projection}, \(n'\) is a trap space of \(f^{\perm}\) where \(n'(v) = n(v)\) for every \(v \in \var{f}\).
	Since \(n(v) = m'(v)\) for every \(v \in \proj\), we have \(n' <_s m\), which contradicts to the \(\leq_s\)-minimality of \(m\) w.r.t. \(f^{\perm}\).
	Hence, \(m'\) is a minimal trap space of \(g\).
	In addition, since \(\phen\) only contains the variables in \(\var{f}\), it is trivial that \(m' \models \phen\).
	
	Second, we consider a minimal trap space \(m\) of \(g\) such that \(m \models \phen\).
	By~\Cref{lem:transformed-BN-mts-fixed-value}, \(m(v) \neq \star\) for every \(v \in \proj\).
	The case of \(m(v^k) = m(v^o) = 1\) is impossible because if it holds, then \(m(g_{v^o}) = m(v^o \land \neg v^k) = 0 \leq_s m(v^o) = 1\), which is a contradiction.
	Let \(\perm\) be the perturbation of \(f\) such that \(\perm(v) = 1\) if and only if \(m(v^k) = 0\) and \(m(v^o) = 1\), \(\perm(v) = 0\) if and only if \(m(v^k) = 1\) and \(m(v^o) = 0\), and \(\perm(v) = \star\) if and only if \(m(v^k) = 0\) and \(m(v^o) = 0\).
	Let \(m'\) be a sub-space of \(f\) such that \(m'(v) = m(v)\) for every \(v \in \var{f}\).
	We show that \(m'\) is a minimal trap space of \(f^{\perm}\) and \(m' \models \phen\) (2).
	
	By Lemma~\ref{lem:trap-space-projection}, \(m'\) is a trap space of \(f^{\perm}\).
	Assume that \(m'\) is not minimal.
	Then there is a minimal trap space \(n\) of \(f^{\perm}\) such that \(n <_s m'\).
	Let \(n'\) be a sub-space of \(g\) such that \(n'(v) = m(v)\) for every \(v \in \proj\) and \(n'(v) = n(v)\) for every \(v \in \var{f}\).
	By following the same reasoning for (1), we have \(n'\) is a trap space of \(g\).
	However, \(n' <_s m\), which contradicts to the \(\leq_s\)-minimality of \(m\) w.r.t. \(g\).
	Hence, \(m'\) is a minimal trap space of \(f^{\perm}\).
	In addition, since \(\phen\) only contains the variables in \(\var{f}\), it is trivial that \(m' \models \phen\).
	
	From (1) and (2), we can conclude that, given  $f, \pert$, and $\phen$, the result of the counting problem \(\acthirdmts{}\) is equivalent to the number of minimal trap spaces of \(g\) that satisfy \(\phen\) where multiple minimal trap spaces with the same values on the variables in \(\proj\) are only counted once.
	By~\Cref{theo:encoding-correctness-mts}, the answer sets of \(\tsconjp{g} \cup \toaspalgname{\phen}\) one-to-one correspond to the minimal trap spaces of \(g\) satisfying \(\phen\).
	The set \(\aproj\) includes the atoms of \(\tsconjp{g} \cup \toaspalgname{\phen}\) corresponding to the variables in \(\proj\) of \(g\).
	It follows that the number of answer sets of \(\tsconjp{g} \cup \toaspalgname{\phen}\) projected to \(\aproj\) is equal to the number of minimal trap spaces of \(g\) satisfying \(\phen\) projected to \(\proj\).
	This implies that \(\acthirdmts{}\) can be computed as the projected answer set counting query \(\projaspcount{\tsconjp{g} \cup \toaspalgname{\phen}, \aproj}\).
\end{proof}

Finally, we show the formal proof of~\Cref{theo:correctness-perturbation-reduction-FIX}.

\correctnessperturbationreductionFIX*
\begin{proof}
	The proof technique of~\Cref{theo:correctness-perturbation-reduction-MTS} can be similarly extended for \acthirdfix{} with the program $\faspp{g}$.
\end{proof}

%% file: sections/appendix-b.tex
\section{Details of Example ASP Programs}\label{sec:detailed-ASP-programs}

The ASP program \(\tsconjp{f}\) considered in~\Cref{exam:third-MTS-projected} is:
\begin{align*}
	&\pos{a} \vee \ngt{a} \leftarrow \top \quad \quad \pos{a} \leftarrow \pos{a}, \ngt{b} \quad \quad \ngt{a} \leftarrow \aux{1} \quad \quad \aux{1} \leftarrow \ngt{a} \quad \quad \aux{1} \leftarrow \pos{b} \\
	&\pos{b} \vee \ngt{b} \leftarrow \top \\
	&\pos{b} \leftarrow \ngt{b^k}, \aux{2} \quad \quad \aux{2} \leftarrow \pos{b^o}  \quad \quad \aux{2} \leftarrow \pos{a} \\
	&\ngt{b} \leftarrow \aux{3} \quad \quad \aux{3} \leftarrow \pos{b^k} \quad \quad \aux{3} \leftarrow \ngt{b^o}, \ngt{a} \\
	&\pos{b^k} \vee \ngt{b^k} \leftarrow \top \quad \quad \pos{b^k} \leftarrow \pos{b^k} \quad \quad \ngt{b^k} \leftarrow \ngt{b^k} \\
	&\pos{b^o} \vee \ngt{b^o} \leftarrow \top \qquad \pos{b^o} \leftarrow \pos{b^o}, \ngt{b^k} \qquad \ngt{b^o} \leftarrow \aux{4} \quad \aux{4} \leftarrow \ngt{b^o} \quad \aux{4} \leftarrow \pos{b^k}
\end{align*}

\noindent The ASP program \(\toaspalgname{\phen}\) considered in~\Cref{exam:third-MTS-projected} is:
\begin{align*}
	&\bot \leftarrow \pos{a} \qquad \bot \leftarrow \dng{\ngt{a}} \\
	&\bot \leftarrow \pos{b} \qquad \bot \leftarrow \dng{\ngt{b}} \\
\end{align*}

%% file: sections/appendix-c.tex
\section{Details of Benchmark Instances}\label{sec:detailed-benchmark-inputs}

Our benchmarks are based on the $245$ real-world biological models from the BBM dataset~\cite{pastva2023repository} available at the time of writing, plus 400 random Boolean networks that were first tested as part of~\cite{TBPS2024}. The real-world Boolean models range up to 1076 variables, out of which up to 223 are source variables. However, the median network size in this dataset is less than 100 variables. As such, we also consider larger, random Boolean networks ranging between 1,000 and 5,000 variables. 
All source variables across all networks are left unrestricted, meaning they can take the value 0 or 1, thereby maximizing the number of admissible trap spaces.

These 645 instances are used directly as inputs for benchmarking the \acfirstfix{} and \acfirstmts{} problems. To test \acsecondfix{} and \acsecondmts{}, we augment each network with a pseudo-random phenotype specification. Here, the specification is chosen as follows: we first compute an arbitrary, fixed minimal trap space using \tsconj. We then randomly select the values of three fixed variables --- excluding all the source variables of the network. 
The conjunction of these values represents the tested phenotype.
This process ensures that for each Boolean network, problem \acsecondmts{} always has at least one valid minimal trap space solution (existence of a fixed point solution cannot be guaranteed regardless of the chosen phenotype).

Finally, to evaluate \acthirdfix{} and \acthirdmts{}, we use the same phenotype but also pseudo-randomly select up to 50 perturbable variables, excluding both the source variables and those fixed by the phenotype. For networks with less than 50 such candidates, we simply select all viable variables as perturbable. We then use the transformation proposed in~\Cref{def:BN-perturbation-trans} to construct a new variant of each benchmark network in which the selected variables can be perturbed. Each such perturbed network, together with the phenotype, represents the input for \acthirdfix{} and \acthirdmts{}.

%% file: sections/appendix-d.tex
\section{Comparison on the Numbers of Solutions and Solved Instances}\label{sec:comparison-num-solved}

The comparison is presented in Figure~\ref{fig:number_of_soln}.

\begin{figure*}
    \centering
    \begin{subfigure}[t]{0.45\textwidth}
        \centering
        \includegraphics[width=0.95\linewidth]{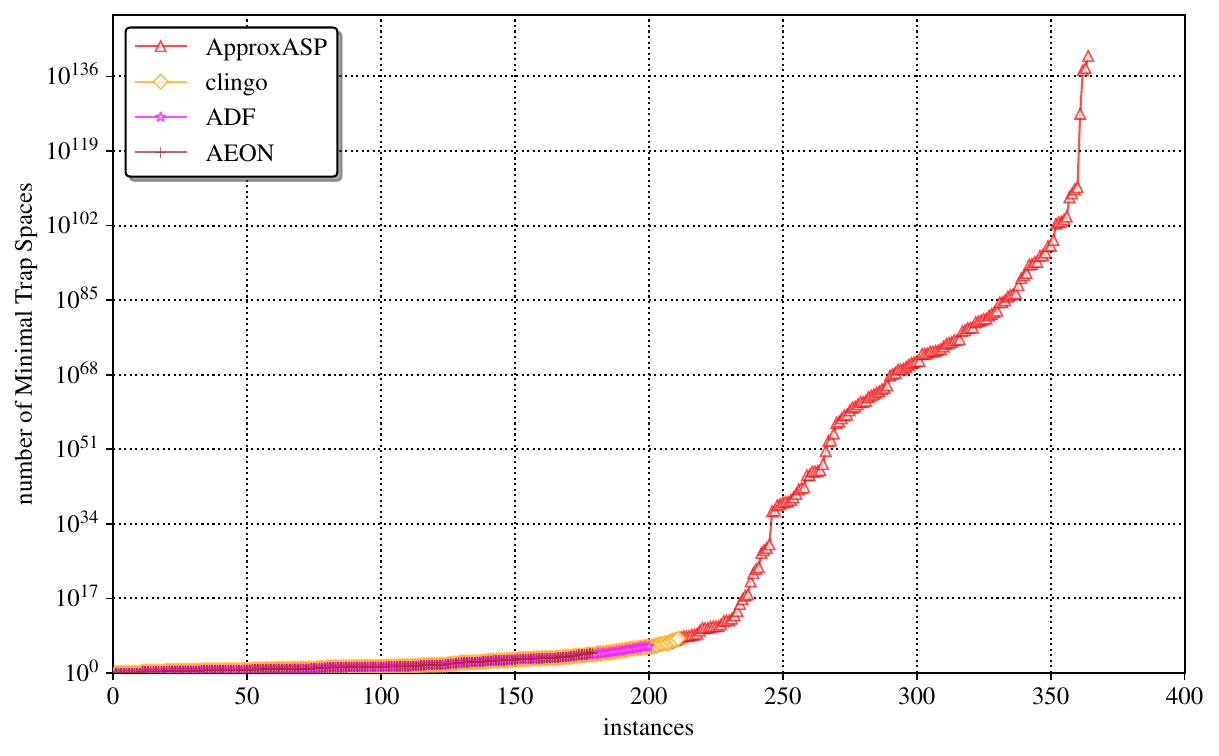}
        \caption{\acfirstmts{}}
    \end{subfigure}
    \begin{subfigure}[t]{0.45\textwidth}
        \centering
        \includegraphics[width=0.95\linewidth]{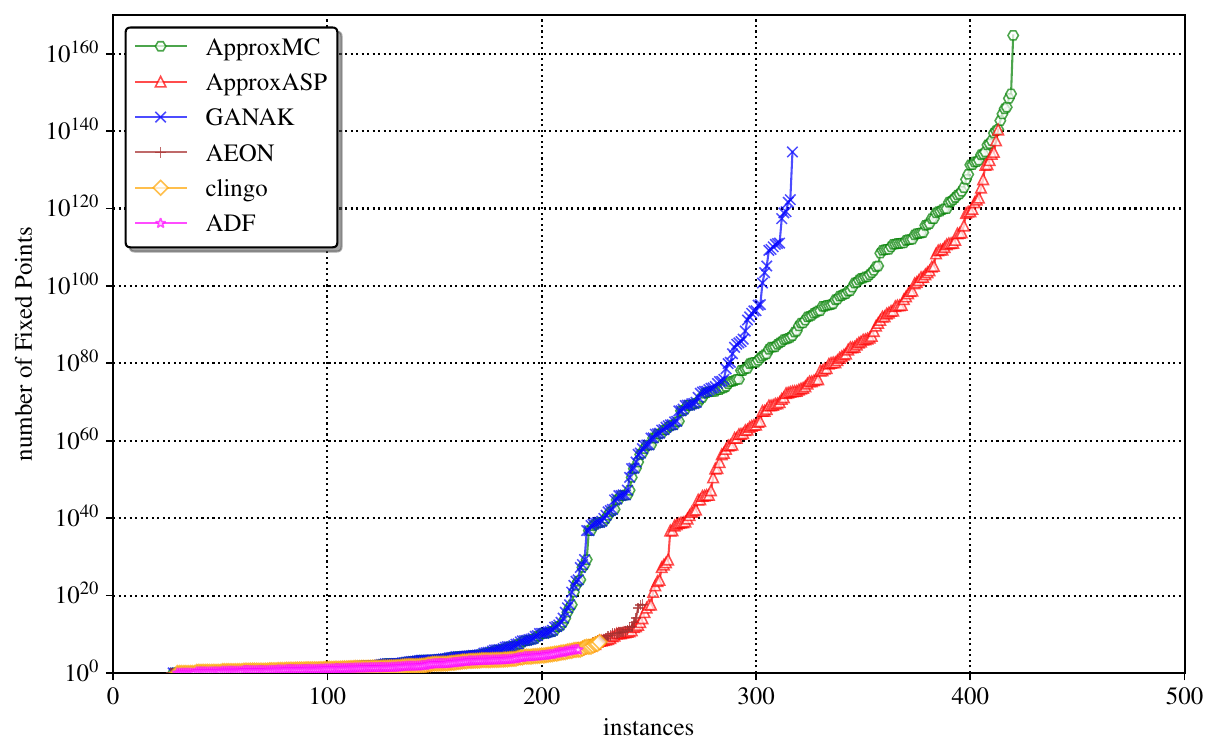}
        \caption{\acfirstfix{}}
    \end{subfigure}
    \begin{subfigure}[t]{0.45\textwidth}
        \centering
        \includegraphics[width=0.95\linewidth]{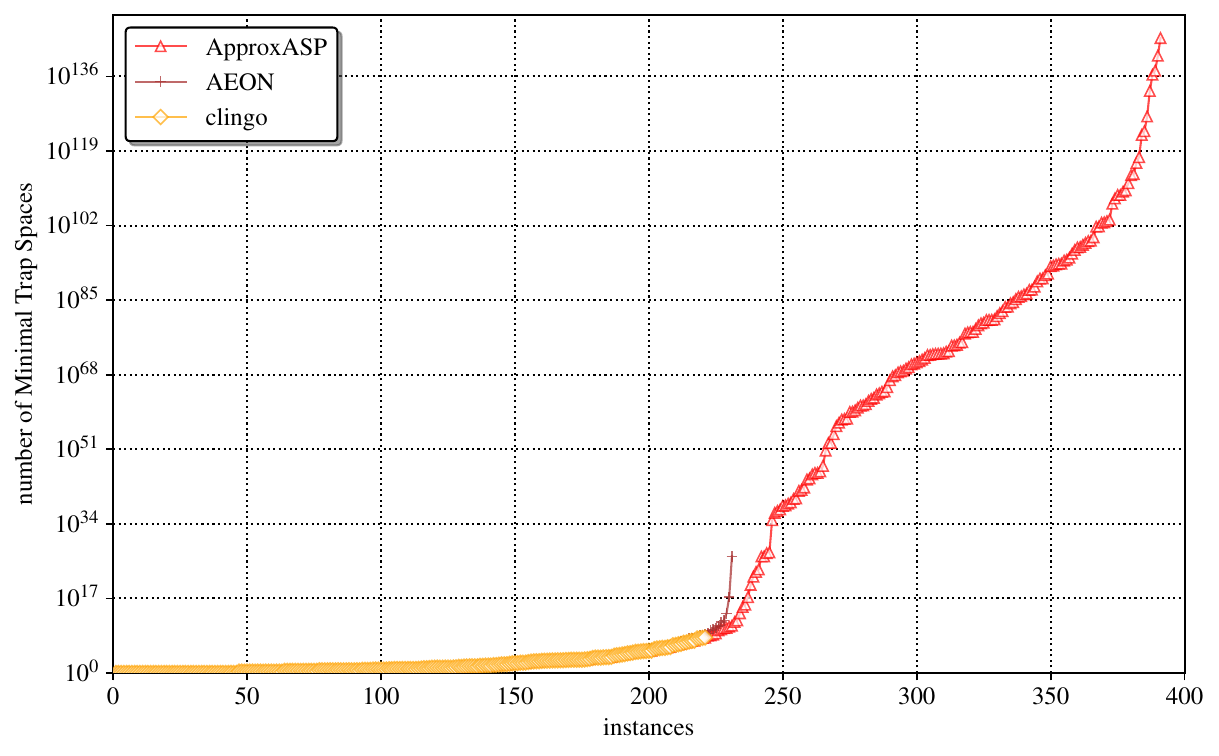}
        \caption{\acsecondmts{}}
    \end{subfigure}
    \begin{subfigure}[t]{0.45\textwidth}
        \centering
        \includegraphics[width=0.95\linewidth]{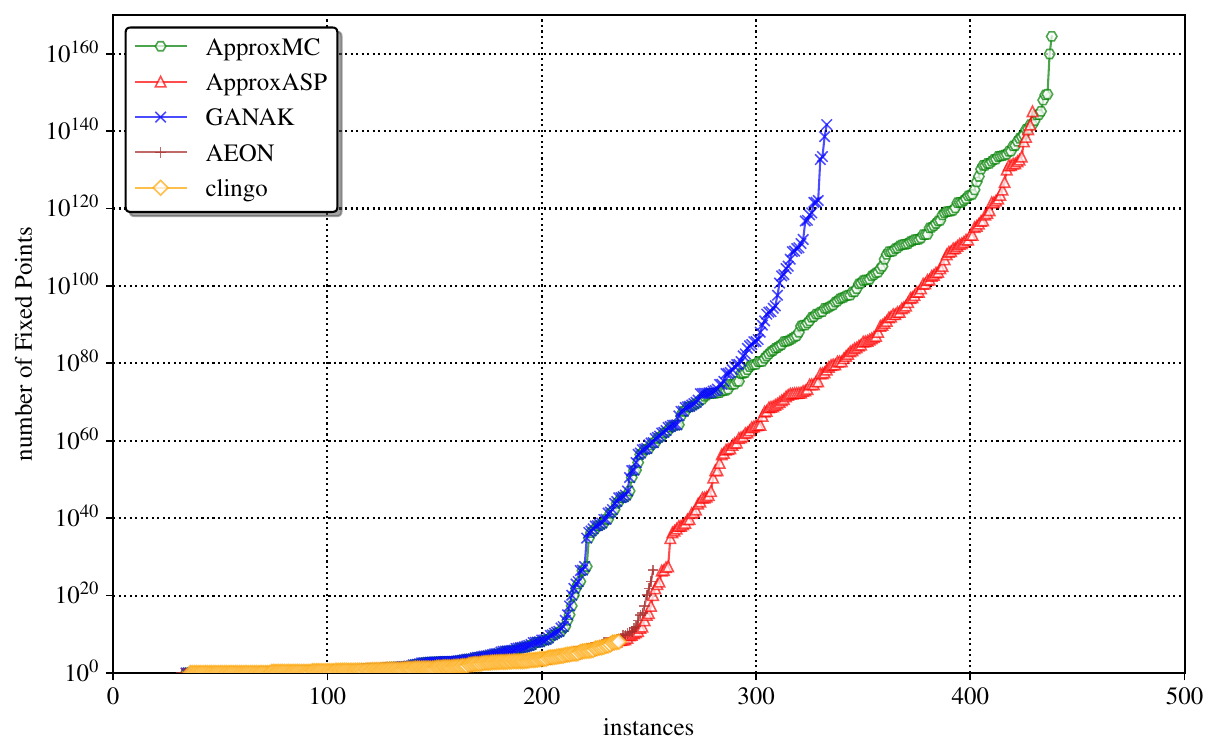}
        \caption{\acsecondfix{}}
    \end{subfigure}
    \begin{subfigure}[t]{0.45\textwidth}
        \centering
        \includegraphics[width=0.95\linewidth]{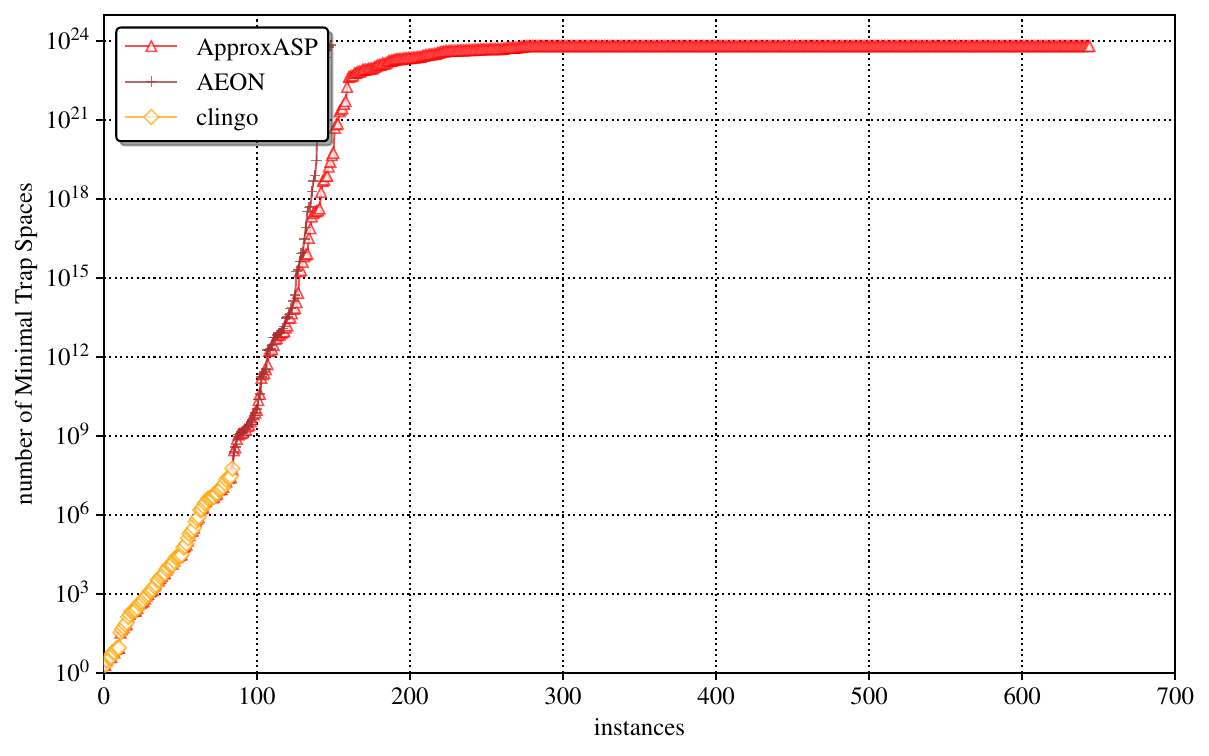}
        \caption{\acthirdmts{}}
    \end{subfigure}
    \begin{subfigure}[t]{0.45\textwidth}
        \centering
        \includegraphics[width=0.95\linewidth]{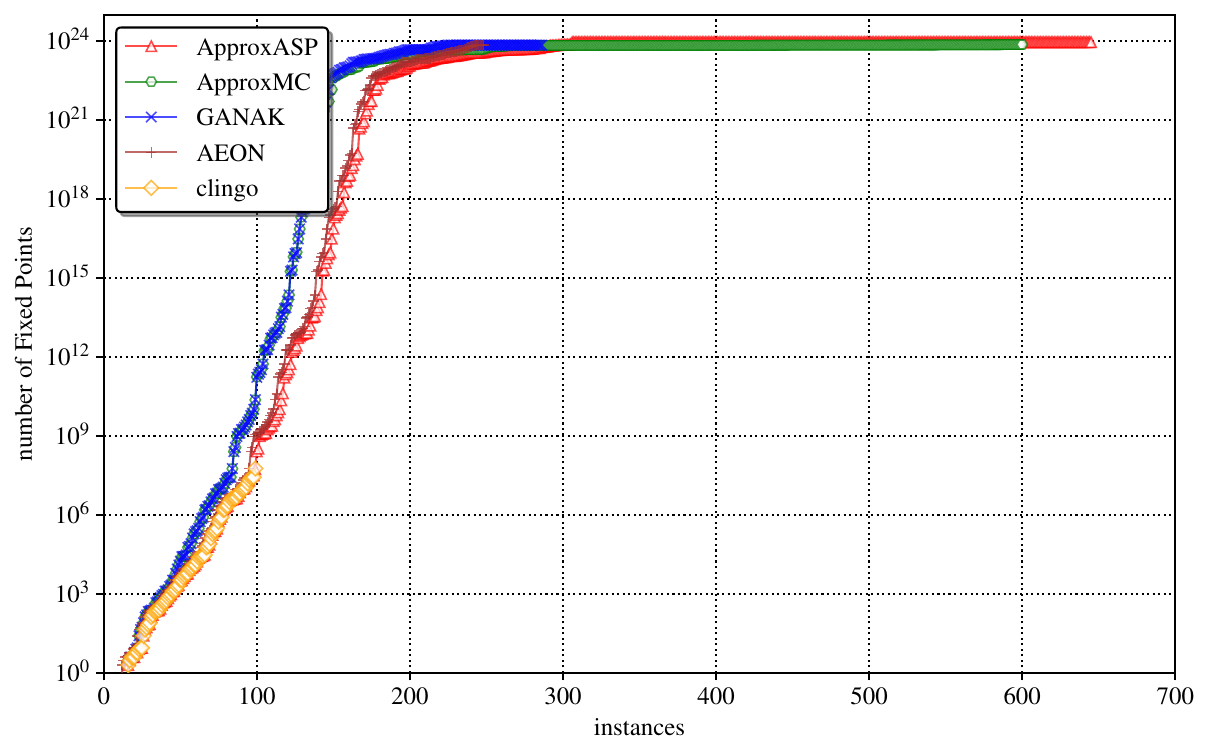}
        \caption{\acthirdfix{}}
    \end{subfigure}
    \caption{The comparison of the number of solutions (or minimal trap spaces and fixed points) for instances solved by different counters.
    A point $(x,y)$ on a plot indicates that a tool can count $x$ instances and each of these instances has at most $y$ solutions.}
    \label{fig:number_of_soln}
\end{figure*}

%% file: sections/appendix-e.tex
\section{Phenotype robustness analysis, case study of Interferon 1 model}
\label{sec:case-study}

Model no. 118 in the BBM dataset~\cite{pastva2023repository} represents an interaction network related to the activation of the so-called \emph{Interferon 1}, a biochemical species closely tied to immune response present in T-cells. The model was initially derived using~\cite{SVATSAJ2020} and then later tuned by domain experts to correctly capture relevant biological phenotypes. It consists of 121 variables, of which 55 are ``inputs'', meaning they are not regulated by other variables.

The model defines three phenotype variables, \emph{ISG} (full name \emph{ISG expression antiviral response phenotype}), \emph{PCK} (full name \emph{Proinflammatory cytokine expression inflammation phenotype}), and \emph{IFN} (full name \emph{Type 1 IFN response phenotype}). As these are separate output variables of the network, each trap space can exhibit any combination of active and inactive phenotype variables.

Since the model defines response of T-cells to immunological stimuli and environmental factors, it is important to understand how these mechanisms respond to potential permanent perturbations, either due to genetic mutations or therapeutic treatments. Here, we provide an overview of the model phenotypes through the lens of phenotype robustness.

For simplicity, we selected 20 variables of the model as potential perturbation targets. In reality, this choice would be further influenced by known genetic risk factors or drug targets. Consequently, this choice results in $3^{20} = 3486784401$ admissible perturbations in our Boolean system. We then consider two phenotype variants: First, where a single phenotype variable is expected to be $1$ and the remaining phenotype variables are unconstrained (i.e. they can be $0$, $1$, or $\star$), yielding three combinations. And second, more specific one, where each phenotype variable is fixed to either $1$ or $0$, yielding eight combinations.

The results for C-MTS-3 and the subsequent robustness computation are given in Table~\ref{table:perturbation-robustness}. Here, we can notice several biologically interesting outcomes:

\begin{itemize}
	\item Regardless of perturbation, the network always exhibits a trap space with the \emph{ISG} phenotype. The remaining phenotypes are usually also present ($r = 0.606$ and $r = 0.663$), but are significantly less robust than \emph{ISG}. This indicates that (assuming favorable environmental conditions), expression of ISG as a response to viral activity is robust and cannot be disrupted by a perturbation.
	\item Among the fully defined phenotypes, $010$ and $000$ are the least robust while $111$ is the most robust. This shows the general tendency of T-cells to reliably and consistently respond to immunological stimuli, as the $111$ phenotype indicates the maximal level of immunological activity.
	\item Even though \emph{ISG} is the most robust phenotype when taken in isolation, phenotypes with active \emph{IFN} generally achieve higher robustness when other phenotypes are required to be inactive.
\end{itemize}

Such outcomes serve several functions: First, they can be used to validate (or refute) assumptions about the biological tendencies of the studied system. Second, if observations do not match model predictions, better quantitative understanding can provide possible sources of further model refinement. Third, this knowledge can enable us to better (more reliably or efficiently) select a phenotype that should be targeted by a treatment among several related but distinct cellular phenotypes. Finally, note that the number of solutions in each case is significantly higher than what would be countable using standard enumeration, underscoring the importance of dedicated counting methods, and approximate counting in particular.

\begin{table}
	\centering
	\caption{Phenotype robustness analysis for the Interferon 1 model. First three columns describe the desired phenotype ($-$ meaning the value is unconstrained). The C-MTS-3 column lists the number of perturbations for which said phenotype appears in the network. Finally, robustness $r$ indicates what portion of the possible perturbations still enable the corresponding phenotype.}
	\begin{tabular}{c | c | c | c | c}
		ISG & PCK & IFN & C-MTS-3 & Robustness ($r$) \\\hline
		1 & - & - & 3486784401 & 1.000 \\
		- & 1 & - & 2114072298 & 0.606 \\
		- & - & 1 & 2313362673 & 0.663 \\\hline
		0 & 0 & 0 & 478296900 & 0.137 \\
		1 & 0 & 0 & 621785970 & 0.178 \\
		0 & 1 & 0 & 478296900 & 0.137 \\	
		0 & 0 & 1 & 813104730 & 0.233 \\	
		1 & 1 & 0 & 782989740 & 0.224 \\
		1 & 0 & 1 & 1096362783 & 0.314 \\
		0 & 1 & 1 & 813104730 & 0.233 \\
		1 & 1 & 1 & 1409735826 & 0.404 \\
	\end{tabular}
	\label{table:perturbation-robustness}
\end{table}

%% file: cp25.bbl
\begin{thebibliography}{10}

\bibitem{AFRM2017}
Emna~Ben Abdallah, Maxime Folschette, Olivier~F. Roux, and Morgan Magnin.
\newblock {ASP}-based method for the enumeration of attractors in
  non-deterministic synchronous and asynchronous multi-valued networks.
\newblock {\em Algorithms Mol. Biol.}, 12(1):20:1--20:23, 2017.

\bibitem{SVATSAJ2020}
Sara~Sadat Aghamiri, Vidisha Singh, Aur{\'{e}}lien Naldi, Tom{\'{a}}s Helikar,
  Sylvain Soliman, Anna Niarakis, and Jinbo Xu.
\newblock Automated inference of {Boolean} models from molecular interaction
  maps using {CaSQ}.
\newblock {\em Bioinform.}, 36(16):4473--4482, 2020.

\bibitem{ADFPR2023}
Mario Alviano, Carmine Dodaro, Salvatore Fiorentino, Alessandro Previti, and
  Francesco Ricca.
\newblock {ASP} and subset minimality: Enumeration, cautious reasoning and
  {MUSes}.
\newblock {\em Artif. Intell.}, 320:103931, 2023.

\bibitem{DBLP:conf/aaai/AzizCMS15}
Rehan~Abdul Aziz, Geoffrey Chu, Christian~J. Muise, and Peter~James Stuckey.
\newblock Stable model counting and its application in probabilistic logic
  programming.
\newblock In {\em {AAAI}}, pages 3468--3474. {AAAI} Press, 2015.

\bibitem{ASDO2024}
Eugenio Azpeitia, Stan~Mu{\~{n}}oz Guti{\'{e}}rrez, David~A. Rosenblueth, and
  Octavio Zapata.
\newblock Bridging abstract dialectical argumentation and {Boolean} gene
  regulation.
\newblock {\em CoRR}, abs/2407.06106, 2024.
\newblock \href {http://arxiv.org/abs/2407.06106} {\path{arXiv:2407.06106}}.

\bibitem{SAT2017}
Tomas Balyo, Marijn~JH Heule, and Matti J{\"a}rvisalo.
\newblock {SAT} competition 2017--solver and benchmark descriptions, 2017.

\bibitem{BD1994}
Rachel Ben{-}Eliyahu and Rina Dechter.
\newblock Propositional semantics for disjunctive logic programs.
\newblock {\em Ann. Math. Artif. Intell.}, 12(1-2):53--87, 1994.

\bibitem{BKPS2020}
Nikola Benes, Lubos Brim, Jakub Kadlecaj, Samuel Pastva, and David
  Safr{\'{a}}nek.
\newblock {AEON:} attractor bifurcation analysis of parametrised {Boolean}
  networks.
\newblock In {\em {CAV}}, pages 569--581. Springer, 2020.

\bibitem{BBPSS2023}
Nikola Benes, Lubos Brim, Samuel Pastva, David Safr{\'{a}}nek, and Eva
  Smij{\'{a}}kov{\'{a}}.
\newblock Phenotype control of partially specified {Boolean} networks.
\newblock In {\em {CMSB}}, pages 18--35. Springer, 2023.

\bibitem{BHV2009}
Armin Biere, Marijn Heule, Hans van Maaren, and Toby Walsh, editors.
\newblock {\em {Handbook} of {Satisfiability}}, volume 185 of {\em Frontiers in
  Artificial Intelligence and Applications}.
\newblock {IOS} Press, 2009.

\bibitem{bloomingdale2018boolean}
Peter Bloomingdale, Van~Anh Nguyen, Jin Niu, and Donald~E Mager.
\newblock Boolean network modeling in systems pharmacology.
\newblock {\em J. Pharmacokinet. Pharmacodyn.}, 45:159--180, 2018.

\bibitem{FAKA2022}
Florian Bridoux, Am{\'{e}}lia Durbec, K{\'{e}}vin Perrot, and Adrien Richard.
\newblock Complexity of fixed point counting problems in {Boolean} networks.
\newblock {\em J. Comput. Syst. Sci.}, 126:138--164, 2022.

\bibitem{CSV2013}
Supratik Chakraborty, Kuldeep~S. Meel, and Moshe~Y. Vardi.
\newblock A scalable approximate model counter.
\newblock In Christian Schulte, editor, {\em {CP}}, pages 200--216. Springer,
  2013.

\bibitem{CMV2016}
Supratik Chakraborty, Kuldeep~S. Meel, and Moshe~Y. Vardi.
\newblock Algorithmic improvements in approximate counting for probabilistic
  inference: From linear to logarithmic {SAT} calls.
\newblock In {\em {IJCAI}}, pages 3569--3576. {IJCAI/AAAI} Press, 2016.

\bibitem{CM78}
Ashok~K. Chandra and George Markowsky.
\newblock On the number of prime implicants.
\newblock {\em Discrete Math.}, 24(1):7–11, 1978.

\bibitem{SVLAL2020}
St{\'{e}}phanie Chevalier, Vincent No{\"{e}}l, Laurence Calzone, Andrei~Yu.
  Zinovyev, and Lo{\"{\i}}c Paulev{\'{e}}.
\newblock Synthesis and simulation of ensembles of {Boolean} networks for cell
  fate decision.
\newblock In {\em {CMSB}}, pages 193--209. Springer, 2020.

\bibitem{clark1978}
Keith~L. Clark.
\newblock {\em Negation as Failure}, page 293–322.
\newblock Springer US, 1978.

\bibitem{DFGH2022}
Ridhwan Dewoprabowo, Johannes~Klaus Fichte, Piotr~Jerzy Gorczyca, and Markus
  Hecher.
\newblock A practical account into counting {Dung's} extensions by dynamic
  programming.
\newblock In {\em {LPNMR}}, pages 387--400. Springer, 2022.

\bibitem{DWM2024}
Yannis Dimopoulos, Wolfgang Dvor{\'{a}}k, and Matthias K{\"{o}}nig.
\newblock Connecting abstract argumentation and {Boolean} networks.
\newblock In {\em {COMMA}}, pages 85--96. {IOS} Press, 2024.

\bibitem{EM2011}
Elena Dubrova and Maxim Teslenko.
\newblock A {SAT}-based algorithm for finding attractors in synchronous
  {Boolean} networks.
\newblock {\em {IEEE} {ACM} Trans. Comput. Biol. Bioinform.}, 8(5):1393--1399,
  2011.

\bibitem{EHK2024}
Thomas Eiter, Markus Hecher, and Rafael Kiesel.
\newblock aspmc: New frontiers of algebraic answer set counting.
\newblock {\em Artif. Intell.}, 330:104109, 2024.

\bibitem{Fages1994}
Fran\c{c}ois Fages.
\newblock Consistency of {Clark's} completion and existence of stable models.
\newblock {\em Methods Log. Comput. Sci.}, 1(1):51--60, 1994.

\bibitem{FHN2022}
Johannes~K. Fichte, Markus Hecher, and Mohamed~A. Nadeem.
\newblock Plausibility reasoning via projected answer set counting - a hybrid
  approach.
\newblock In {\em {IJCAI}}, volume~22, page 2620–2626, 2022.

\bibitem{FGR2022}
Johannes~Klaus Fichte, Sarah~Alice Gaggl, and Dominik Rusovac.
\newblock Rushing and strolling among answer sets - navigation made easy.
\newblock In {\em {AAAI}}, pages 5651--5659. {AAAI} Press, 2022.

\bibitem{FH2019}
Johannes~Klaus Fichte and Markus Hecher.
\newblock Treewidth and counting projected answer sets.
\newblock In {\em {LPNMR}}, pages 105--119. Springer, 2019.

\bibitem{FHM2024}
Johannes~Klaus Fichte, Markus Hecher, and Arne Meier.
\newblock Counting complexity for reasoning in abstract argumentation.
\newblock {\em J. Artif. Intell. Res.}, 80, 2024.

\bibitem{FHMW2017}
Johannes~Klaus Fichte, Markus Hecher, Michael Morak, and Stefan Woltran.
\newblock Answer set solving with bounded treewidth revisited.
\newblock In {\em {LPNMR}}, pages 132--145. Springer, 2017.

\bibitem{Fischer2021}
Stephan Fischer and Jesse Gillis.
\newblock How many markers are needed to robustly determine a cell’s type?
\newblock {\em iScience}, 24(11):103292, November 2021.

\bibitem{GKKOSS2011}
Martin Gebser, Benjamin Kaufmann, Roland Kaminski, Max Ostrowski, Torsten
  Schaub, and Marius Schneider.
\newblock Potassco: The {Potsdam} answer set solving collection.
\newblock {\em {AI} Commun.}, 24(2):107--124, 2011.

\bibitem{gelfond1988stable}
Michael Gelfond and Vladimir Lifschitz.
\newblock The stable model semantics for logic programming.
\newblock In {\em {ICLP}}, pages 1070--1080. {MIT} Press, 1988.

\bibitem{HMJ2024}
Jesse Heyninck, Matthias Knorr, and Jo{\~{a}}o Leite.
\newblock Abstract dialectical frameworks are {Boolean} networks.
\newblock In {\em {LPNMR}}, pages 98--111. Springer, 2024.

\bibitem{CS2015}
Christopher~M. Homan and Sven Kosub.
\newblock Dichotomy results for fixed point counting in {Boolean} dynamical
  systems.
\newblock {\em Theor. Comput. Sci.}, 573:16--25, 2015.

\bibitem{Inoue2011}
Katsumi Inoue.
\newblock Logic programming for {Boolean} networks.
\newblock In {\em IJCAI}, pages 924--930. {IJCAI/AAAI}, 2011.

\bibitem{IS2012}
Katsumi Inoue and Chiaki Sakama.
\newblock Oscillating behavior of logic programs.
\newblock In {\em Correct Reasoning - Essays on Logic-Based {AI} in Honour of
  Vladimir Lifschitz}, pages 345--362. Springer, 2012.

\bibitem{Janhunen2006}
Tomi Janhunen.
\newblock Some (in)translatability results for normal logic programs and
  propositional theories.
\newblock {\em J. Appl. Non Class. Logics}, 16(1-2):35--86, 2006.

\bibitem{JNSSY2006}
Tomi Janhunen, Ilkka Niemel{\"{a}}, Dietmar Seipel, Patrik Simons, and
  Jia{-}Huai You.
\newblock Unfolding partiality and disjunctions in stable model semantics.
\newblock {\em {ACM} TOCL}, 7(1):1--37, 2006.

\bibitem{KCM2024}
Mohimenul Kabir, Supratik Chakraborty, and Kuldeep~S. Meel.
\newblock Exact {ASP} counting with compact encodings.
\newblock In {\em {AAAI}}, pages 10571--10580. {AAAI} Press, 2024.

\bibitem{KESHFM22}
Mohimenul Kabir, Flavio~O. Everardo, Ankit~K. Shukla, Markus Hecher,
  Johannes~Klaus Fichte, and Kuldeep~S. Meel.
\newblock {ApproxASP} - a scalable approximate answer set counter.
\newblock In {\em {AAAI}}, pages 5755--5764. {AAAI} Press, 2022.

\bibitem{KM2023}
Mohimenul Kabir and Kuldeep~S Meel.
\newblock A fast and accurate {ASP} counting based network reliability
  estimator.
\newblock In {\em LPAR}, volume~94, pages 270--287, 2023.

\bibitem{KM2024}
Mohimenul Kabir and Kuldeep~S Meel.
\newblock On lower bounding minimal model count.
\newblock {\em TPLP}, 24(4):586–605, July 2024.

\bibitem{KBDDGH2022}
Nikolai K{\"{a}}fer, Christel Baier, Martin Diller, Clemens Dubslaff,
  Sarah~Alice Gaggl, and Holger Hermanns.
\newblock Admissibility in probabilistic argumentation.
\newblock {\em J. Artif. Intell. Res.}, 74, 2022.

\bibitem{KSSV2013}
Roland Kaminski, Torsten Schaub, Anne Siegel, and Santiago Videla.
\newblock Minimal intervention strategies in logical signaling networks with
  {ASP}.
\newblock {\em TPLP}, 13(4-5):675--690, 2013.

\bibitem{kitano2007towards}
Hiroaki Kitano.
\newblock Towards a theory of biological robustness.
\newblock {\em Mol. Syst. Biol.}, 3(1):137, 2007.

\bibitem{HAH2015}
Hannes Klarner, Alexander Bockmayr, and Heike Siebert.
\newblock Computing maximal and minimal trap spaces of {Boolean} networks.
\newblock {\em Nat. Comput.}, 14(4):535--544, 2015.

\bibitem{KHNS2018}
Hannes Klarner, Frederike Heinitz, Sarah Nee, and Heike Siebert.
\newblock Basins of attraction, commitment sets, and phenotypes of {Boolean}
  networks.
\newblock {\em {IEEE} {ACM} Trans. Comput. Biol. Bioinform.}, 17(4):1115--1124,
  2020.

\bibitem{KTFJCS2021}
Hannes Klarner, Elisa Tonello, Laura Fontanals, Florence Janody, Claudine
  Chaouiya, and Heike Siebert.
\newblock Detection of markers for discrete phenotypes.
\newblock In {\em {CSBio}}, pages 64--68, 2021.

\bibitem{LL03}
Joohyung Lee and Vladimir Lifschitz.
\newblock Loop formulas for disjunctive logic programs.
\newblock In {\em {ICLP}}, pages 451--465. Springer, 2003.

\bibitem{LQADZXHE2016}
Qin Li, Anders Wennborg, Erik Aurell, Erez Dekel, Jie-Zhi Zou, Yuting Xu, Sui
  Huang, and Ingemar Ernberg.
\newblock Dynamics inside the cancer cell attractor reveal cell heterogeneity,
  limits of stability, and escape.
\newblock {\em Proc. Natl. Acad. Sci. U.S.A.}, 113(10):2672–2677, February
  2016.

\bibitem{LMNWW22}
Thomas Linsbichler, Marco Maratea, Andreas Niskanen, Johannes~Peter Wallner,
  and Stefan Woltran.
\newblock Advanced algorithms for abstract dialectical frameworks based on
  complexity analysis of subclasses and {SAT} solving.
\newblock {\em Artif. Intell.}, 307:103697, 2022.

\bibitem{MT1999}
Victor~W Marek and Miroslaw Truszczy{\'n}ski.
\newblock Stable models and an alternative logic programming paradigm.
\newblock In {\em The Logic Programming Paradigm}, pages 375--398. Springer,
  1999.

\bibitem{MSS2023}
Gabriele Masina, Giuseppe Spallitta, and Roberto Sebastiani.
\newblock On {CNF} conversion for disjoint {SAT} enumeration.
\newblock In {\em {SAT}}, pages 15:1--15:16, 2023.

\bibitem{montagud2022patient}
Arnau Montagud, Jonas B{\'e}al, Luis Tobalina, Pauline Traynard, Vigneshwari
  Subramanian, Bence Szalai, R{\'o}bert Alf{\"o}ldi, L{\'a}szl{\'o} Pusk{\'a}s,
  Alfonso Valencia, Emmanuel Barillot, et~al.
\newblock Patient-specific {Boolean} models of signalling networks guide
  personalised treatments.
\newblock {\em Elife}, 11:e72626, 2022.

\bibitem{Padoa1901}
Alessandro Padoa.
\newblock Essai d'une théorie algébrique des nombres entiers, précédé
  d'une introduction logique à une théorie déductive quelconque.
\newblock {\em Bibliothèque du Congrès International de Philosophie},
  3:309--365, 1901.

\bibitem{pastva2023repository}
Samuel Pastva, David {\v S}afr{\'a}nek, Nikola Bene{\v s}, Lubo{\v s} Brim, and
  Thomas Henzinger.
\newblock Repository of logically consistent real-world {Boolean} network
  models.
\newblock {\em bioRxiv}, 2023.
\newblock URL:
  \url{https://www.biorxiv.org/content/early/2023/06/12/2023.06.12.544361}.

\bibitem{Paulev2020}
Loïc Paulev{\'{e}}, Juraj Kol{\v{c}}{\'{a}}k, Thomas Chatain, and Stefan Haar.
\newblock Reconciling qualitative, abstract, and scalable modeling of
  biological networks.
\newblock {\em Nat. Commun.}, 11(1):1--7, August 2020.

\bibitem{Rozum2021}
Jordan~C. Rozum, Jorge G{\'{o}}mez~Tejeda Za{\~{n}}udo, Xiao Gan, D{\'{a}}vid
  Deritei, and R{\'{e}}ka Albert.
\newblock Parity and time reversal elucidate both decision-making in empirical
  models and attractor scaling in critical {Boolean} networks.
\newblock {\em Sci. Adv.}, 7(29):eabf8124, July 2021.

\bibitem{ST2009}
Torsten Schaub and Sven Thiele.
\newblock Metabolic network expansion with answer set programming.
\newblock In {\em ICLP}, pages 312--326. Springer, 2009.

\bibitem{SKIKK2020}
Julian~D Schwab, Silke~D K{\"u}hlwein, Nensi Ikonomi, Michael K{\"u}hl, and
  Hans~A Kestler.
\newblock Concepts in {Boolean} network modeling: What do they all mean?
\newblock {\em Comput. Struct. Biotechnol. J.}, 18:571--582, 2020.

\bibitem{ST23}
Robert Schwieger and Elisa Tonello.
\newblock Reduction for asynchronous {Boolean} networks: elimination of
  negatively autoregulated components.
\newblock {\em Discret. Math. Theor. Comput. Sci.}, 25(2), 2023.

\bibitem{SRSM19}
Shubham Sharma, Subhajit Roy, Mate Soos, and Kuldeep~S. Meel.
\newblock {GANAK:} {A} scalable probabilistic exact model counter.
\newblock In {\em {IJCAI}}, pages 1169--1176. ijcai.org, 2019.

\bibitem{SDZ2002}
Ilya Shmulevich, Edward~R. Dougherty, and Wei Zhang.
\newblock From {Boolean} to probabilistic {Boolean} networks as models of
  genetic regulatory networks.
\newblock {\em Proc. {IEEE}}, 90(11):1778--1792, 2002.

\bibitem{shmulevich2002gene}
Ilya Shmulevich, Edward~R Dougherty, and Wei Zhang.
\newblock Gene perturbation and intervention in probabilistic {Boolean}
  networks.
\newblock {\em Bioinf.}, 18(10):1319--1331, 2002.

\bibitem{su2020sequential}
Cui Su and Jun Pang.
\newblock Sequential temporary and permanent control of {Boolean} networks.
\newblock In {\em {CMSB}}, pages 234--251. Springer, 2020.

\bibitem{Thurley2006}
Marc Thurley.
\newblock {sharpSAT} - counting models with advanced component caching and
  implicit {BCP}.
\newblock In {\em {SAT}}, pages 424--429. Springer, 2006.

\bibitem{TP2024}
Elisa Tonello and Loïc Paulevé.
\newblock Phenotype control and elimination of variables in {Boolean} networks.
\newblock {\em Peer Community Journal}, 4, August 2024.

\bibitem{Predrag2006}
Predrag~T. Tosic.
\newblock On the complexity of counting fixed points and gardens of eden in
  sequential dynamical systems on planar bipartite graphs.
\newblock {\em Int. J. Found. Comput. Sci.}, 17(5):1179--1204, 2006.

\bibitem{PG2005}
Predrag~T. Tosic and Gul~A. Agha.
\newblock On computational complexity of counting fixed points in symmetric
  {Boolean} graph automata.
\newblock In {\em {UC}}, pages 191--205. Springer, 2005.

\bibitem{TBPS2024}
Van{-}Giang Trinh, Belaid Benhamou, Samuel Pastva, and Sylvain Soliman.
\newblock Scalable enumeration of trap spaces in {Boolean} networks via answer
  set programming.
\newblock In {\em {AAAI}}, pages 10714--10722. {AAAI} Press, 2024.

\bibitem{VML2024}
Van{-}Giang Trinh, Belaid Benhamou, and Lo{\"{\i}}c Paulev{\'{e}}.
\newblock mpbn: a simple tool for efficient edition and analysis of elementary
  properties of {Boolean} networks.
\newblock {\em CoRR}, abs/2403.06255, 2024.
\newblock \href {http://arxiv.org/abs/2403.06255} {\path{arXiv:2403.06255}}.

\bibitem{TBV2025}
Van{-}Giang Trinh, Belaid Benhamou, and Vincent Risch.
\newblock Graphical analysis of abstract argumentation frameworks via {Boolean}
  networks.
\newblock In {\em {ICAART}}, pages 745--756, 2025.

\bibitem{TBS2023}
Van{-}Giang Trinh, Belaid Benhamou, and Sylvain Soliman.
\newblock Efficient enumeration of fixed points in complex {Boolean} networks
  using answer set programming.
\newblock In {\em {CP}}, pages 35:1--35:19. Schloss Dagstuhl - Leibniz-Zentrum
  f{\"{u}}r Informatik, 2023.

\bibitem{Trinh2023}
Van-Giang Trinh, Belaid Benhamou, and Sylvain Soliman.
\newblock Trap spaces of {Boolean} networks are conflict-free siphons of their
  {Petri} net encoding.
\newblock {\em Theor. Comput. Sci.}, 971:114073, September 2023.

\bibitem{TBSF2024}
Van{-}Giang Trinh, Belaid Benhamou, Sylvain Soliman, and Fran\c{c}ois Fages.
\newblock Graphical conditions for the existence, unicity and number of regular
  models.
\newblock In {\em {ICLP}}, pages 175--187, 2024.

\bibitem{TKB2022}
Van{-}Giang Trinh, Kunihiko Hiraishi, and Belaid Benhamou.
\newblock Computing attractors of large-scale asynchronous {Boolean} networks
  using minimal trap spaces.
\newblock In {\em {ACM-BCB}}, pages 13:1--13:10. {ACM}, 2022.

\bibitem{TPPR2024}
Van-Giang Trinh, Kyu~Hyong Park, Samuel Pastva, and Jordan~C Rozum.
\newblock Mapping the attractor landscape of {Boolean} networks with biobalm.
\newblock {\em Bioinformatics}, 41(5):btaf280, 2025.

\bibitem{VGETGNSSS2015}
Santiago Videla, Carito Guziolowski, Federica Eduati, Sven Thiele, Martin
  Gebser, Jacques Nicolas, Julio Saez{-}Rodriguez, Torsten Schaub, and Anne
  Siegel.
\newblock Learning {Boolean} logic models of signaling networks with {ASP}.
\newblock {\em Theor. Comput. Sci.}, 599:79--101, 2015.

\bibitem{YM2023}
Jiong Yang and Kuldeep~S Meel.
\newblock Rounding meets approximate model counting.
\newblock In {\em CAV}, pages 132--162. Springer, 2023.

\end{thebibliography}
